\documentclass[a4paper,UKenglish,cleveref, autoref, thm-restate]{lipics-v2021}

\usepackage{xcolor}
\usepackage{graphicx}
\usepackage{tikz}
\hideLIPIcs
\usetikzlibrary {shapes.geometric}
\usetikzlibrary{decorations.pathmorphing, positioning}
\nolinenumbers
\usepackage{bm}
\usepackage{amssymb}
\usepackage{hyperref}
\usepackage{wrapfig}
\usepackage[normalem]{ulem}
\usepackage{algpseudocode}
\usepackage{algorithm}
\algnewcommand\algorithmicforeach{\textbf{for each}}
\algdef{S}[FOR]{ForEach}[1]{\algorithmicforeach\ #1\ \algorithmicdo}
\newtheorem{notation}{Notation}
\usepackage{apxproof}
\title{Null Messages, Information and Coordination}

\author{Ra\"issa Nataf}{Technion, Israel} {raissa.nataf@campus.technion.ac.il}{}{}
\author{Guy Goren}{Protocol Labs, Israel} {guy.goren@protocol.ai}{}{}
\author{Yoram Moses}{Technion, Israel} {moses@ee.technion.ac.il}{}{ Yoram Moses is the Israel Pollak academic chair at the Technion. Both his work and that of Ra\"issa Nataf were supported in part by the Israel Science Foundation under grant 2061/19.}
\authorrunning{R.Nataf, G.Goren, Y.Moses}
\Copyright{Ra\"issa Nataf, Guy Goren and Yoram Moses}
\ccsdesc{Theory of computation~Distributed algorithms}
\ccsdesc{Computing methodologies~Reasoning about belief and knowledge}

\Copyright{Ra\"issa Nataf, Guy Goren and Yoram Moses}

\keywords{null messages, fault tolerance, coordination, information flow, knowledge analysis.}

\relatedversion{} 

\nolinenumbers 

\newcommand{\Btagnullfree}{B'_{\not{n}}}
\newcommand{\CG}{\mathsf{CG}}
\newcommand{\nG}{\mathsf{nG}}
\newcommand{\ba}{\underline{a}} 

\newcommand{\Bnullfree}{B_{\not{n}}}
\newcommand{\Eactual}{E_{\sf a}}
\newcommand{\Enull}{E_{\sf n}}
\newcommand{\Eloc}{E_{\sf l}}
\newcommand{\runfact}{\psi}
\newcommand{\nicerun}{\psi_{nice}}

\newcommand{\node}[1]{\langle#1\rangle}

\newcommand{\NbM}{\mbox{$\mathsf{NbM}$}}
\newcommand{\RbM}{\mbox{$\mathsf{RbM}$}}

\newcommand{\chan}[1]{\mathsf{ch}_{#1}}
\newcommand{\rnice}{\hat{r}}

\newcommand{\comm}[1]{}
\newcommand{\Proc}{\mathbb{P}}

\newcommand{\nodes}{\mathbb{V}}

\newcommand{\faulty}{\mathbb{F}}
\newcommand{\modelf}{\gamma^f}
\newcommand{\defemph}[1]{\textbf{\textit{#1}}}

\newcommand{\act}{\alpha}
\newcommand{\RP}{R_Q}

\newcommand{\factOR}{v_s=1}

\newcommand{\fF}{f/\mathsf{failed}}
\newcommand{\OR}{{\sf O-R}}

\newcommand{\ORi}{\mbox{$\mathsf{OR}$}}
\newcommand{\IT}{\mbox{$\mathsf{IT}$}}

\newcommand{\Nat}{{\mathbb{N}}}

\newcommand{\sat}{\vDash}

\newcommand{\angles}[1]{\ensuremath{\node{{#1}}}}

\newcommand{\emc}{\rightsquigarrow}

\begin{document}

\maketitle

\begin{abstract}
This paper investigates the role that null messages play in synchronous systems with and without failures, and provides necessary and sufficient conditions on the structure of protocols for information transfer and coordination there. 
We start by introducing  a new and more refined definition of null messages. A generalization of message chains that allow these null messages is provided,  and is shown to be necessary and sufficient for information transfer in reliable systems.   Coping with crash failures requires a much richer structure, since  not receiving a message may be the result of the sender's failure. We introduce a class of communication patterns called {\em resilient message blocks}, which impose a stricter condition on protocols than the {\em silent choirs} of Goren and Moses (2020). Such blocks are shown to be necessary for information transfer in crash-prone systems. Moreover, they are sufficient in several cases of interest, in which silent choirs are not. Finally, a particular combination of resilient message blocks is shown to be necessary and sufficient for solving the Ordered Response coordination problem.
\end{abstract}

\section{Introduction}
Communication and coordination in distributed systems depend crucially on properties of the model at hand. 
In synchronous systems in which processes have clocks and message transmission times are bounded, sending explicit messages is not the only way to transmit information. 
Suppose that a sender~$s$ needs to transmit  its (binary) initial value~$v_s$ to a destination process~$d$, in a system in which messages are delivered in~1 time step.  
If~$s$ follows a protocol by which it sends~$d$ a message at time~0 in case $v_s=0$ and does not send anything if $v_s=1$, then~$d$ can learn that $v_s=1$ at time 1 without receiving any messages. Lamport called this 
``{\it sending  a message by not sending a message}'' in~\cite{lamp}, and he referred to not sending a message over a communication channel at a given time~$t$ as sending a ``{\em null message}.''   
In this paper we provide a new and more precise definition of null messages, and investigate the general role that null messages play in information transfer and in  coordination in synchronous systems with and without failures. 

 In particular, our results extend and generalize those of Goren and Moses in~\cite{silence}, who were the first to explicitly consider how silence can be used in systems with crash failures.

The possibility of failures makes information transfer a rather subtle issue. Denote by~$f$ an {\it a priori} upper bound on the number of failures per execution.
Consider the following protocol, which we denote $P_1$: In the first round process $s$ sends a message to $p$ reporting whether its initial value~$v_s$ is $1$ or $0$. In round~2, if process~$p$ has received a message stating that $v_s=1$ it keeps silent, and it sends an actual message to $d$ otherwise.
A run $r_1$ of $P_1$ in which $v_s=1$ and no process fails is depicted in \Cref{figure1}. (In our figures, solid arrows represent actual messages and dashed arrows represent null messages.) Note that if $f=0$, then process~$p$   receives the first round message from~$s$ in every run of~$P_1$.  Consequently, if $v_s=1$ then following the second round, $d$ learns that $v_s=1$ since it did not hear from~$p$. 
Now assume that one process may crash ($f=1$). In this case $r_1$, where  none fails, is a legal execution of~$P_1$ but~$d$ is not informed that $v_s=1$ in~$r_1$. This is because~$d$ cannot distinguish~$r_1$ 
from a run in which $v_s=0$ and~$p$ crashes before sending its message to~$d$.

\begin{figure}[!ht]
\begin{minipage}[b]{.4\textwidth}
\centering
\includegraphics[width=1\textwidth]
{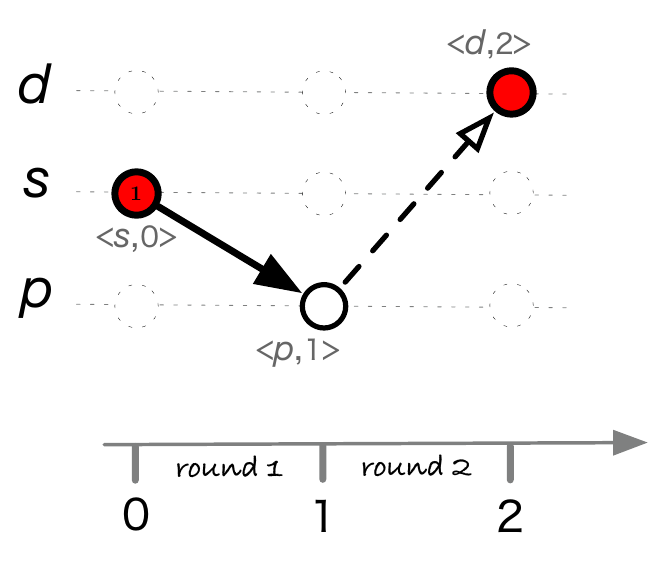}
\caption{The run~$r_1$ of~$P_1$ in which~$d$ is informed that $v_s=1$ when $f=0$ but {\bf not} when $f=1$.}
\label{figure1}
\end{minipage}
\hfill
\begin{minipage}[b]{.5\textwidth}
\centering
\includegraphics[width=.8\textwidth]
{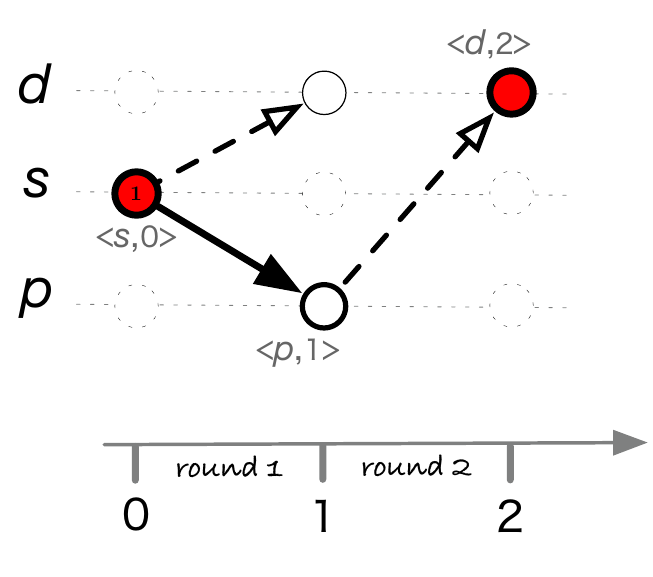}
\caption{The run~$r_2$ of~$P_2$ in which a silent choir informs~$d$ that $v_s=1$ when $f=1$.
\\$~$}
\label{figure2}
\end{minipage}
\end{figure}
Consider now  a protocol $P_2$ that differs from~$P_1$ only in that according to~$P_2$ process~$s$ should send~$d$ a message at time~0 if and only if $v_s=0$. (Process~$s$ remains silent if $v_s=1$.) \Cref{figure2} depicts the run $r_2$ of $P_2$ where $v_s=1$ and no failures occur. Observe that in case $f=1$, process~$d$ is informed in $r_2$ that $v_s=1$. 
As before, $d$ cannot observe in~$r_2$ whether~$p$ has crashed. However, since $f=1$, at most one of the missing messages to~$d$ can be explained by a process crash. The other missing message must be caused by the fact that $v_s=1$. Note that exactly the same messages are sent in $r_1$ and $r_2$. Process~$d$ obtains different information in the two cases because the protocols are different: 
In particular, $s$ keeps silent toward $d$ only under certain conditions of interest according to $P_2$ while it always keeps silent according to $P_1$. This is what provides $d$ genuine information.

The above discussion motivates a new and more refined definition of null messages. While \cite{lamp} considers not receiving a message as the receipt of a null message, we define a null message to be sent by a process~$i$ to its neighbor~$j$ at time~$t$ in a given execution if process~$i$ does not send an actual message at time~$t$, and there is at least one execution of the protocol in which~$i$ {\em does} send~$j$ an actual message at time~$t$. (A formal definition appears in~\cref{sec:model}.) With such a definition, a null message is guaranteed to carry some nontrivial information.

Goren and Moses showed in~\cite{silence} that information can be transmitted in silence even when crashes may occur. Their {\em Silent Choir\,} Theorem states a necessary condition for~$d$ to  learn the initial value of~$s$ in a crash-prone system without a message chain from~$s$.  For failure-free executions, their necessary condition becomes the following: 
If $d$ knows the value of~$v_s$ at time~$m>0$ without an actual message chain (i.e., a message chain exclusively composed of actual messages) from~$s$ having reached it, there are at least $f+1$ processes that receive an actual message chain from~$s$ by time~$m-1$ and send no message to $d$ at time $m-1$. 
These processes are called a {\em silent choir}. In~$r_2$, process~$d$ learns the value of~$s$ at time~2, and we can see in \Cref{figure2} that the set $\{s,p\}$ constitutes a silent choir.  
However, as we shall see, while being necessary for information transfer in crash-prone synchronous systems, the  Silent Choir Theorem's conditions are not sufficient, even for failure-free executions.

Consider again the run~$r_1$ of~\Cref{figure1} and assume $f=1$. The set $\{s,p\}$ forms a silent choir, i.e., the conditions of the Silent Choir Theorem hold. However, as we described above, $d$ does not learn that $v_s=1$ in~$r_1$. 
Process~$s$, who belongs to the silent choir here, {\em never} sends $d$ a message under~$P_1$, and so its silence does not form a null message according to our new definition.
Naturally, strengthening the Silent Choir condition by requiring  that the processes of the silent choir  actually send null messages at time $m-1$, i.e., one time unit before $d$ gains the knowledge that $v_s=1$, could make it sufficient. However, the condition would then \textbf{not} be necessary. For example, as depicted in \Cref{figure2}, $s$ does not send~$d$ a null message in the second round of~$r_2$ (which is at time~$m-1$ in the language of the Silent Choir Theorem).  
Notice that, in addition, members of the silent choir do not necessarily send   null messages to $d$. Indeed, they do not even need to be neighbors of~$d$.

Beyond the fundamental value of studying null messages to understand the differences between synchronous and asynchronous models of distributed systems, judicious use of null messages can lead to considerable savings. 
For a concrete example, consider a network structured as depicted in \Cref{fig:ONet}, in which there are different costs for sending over different channels. This can arise, for example, from  the three intermediate processes residing at the same site or belonging to the same organization as~$s$, while~$d$ is across the ocean, or just connected via expensive channels. Assume, in addition, that $d$ needs to know the value of~$v_s$, where normally $v_s=1$ and only very rarely $v_s\ne 1$. 
Finally, we wish to be able to overcome up to  $\pmb{f=2}$ process crashes. While sending an actual message chain from~$s$ to~$d$ would cost \$1001, null messages can be used to inform~$d$ that $v_s=1$ at a cost of \$2  if $v_s=1$ and no failures occur (see \Cref{fig:ONull}). With such a solution, if  $v_s=0$ then the cost may in the worst case be as high as~\$3003. But if the latter is rare and the former is very common, such use of null messages can provide a clear advantage.   E.g., if for every 100 runs in which $v_s=1$ and no failures occur we expect to see~2 runs where this is not the case, then   using null messages for the 102 runs will cost at most \$6206,  compared to \$102102 spent by a protocol that uses only actual messages in the network of \Cref{fig:ONet}. 

\begin{figure}[ht]
\begin{minipage}[b]{.4\textwidth}
\centering
\includegraphics[width=1\textwidth]{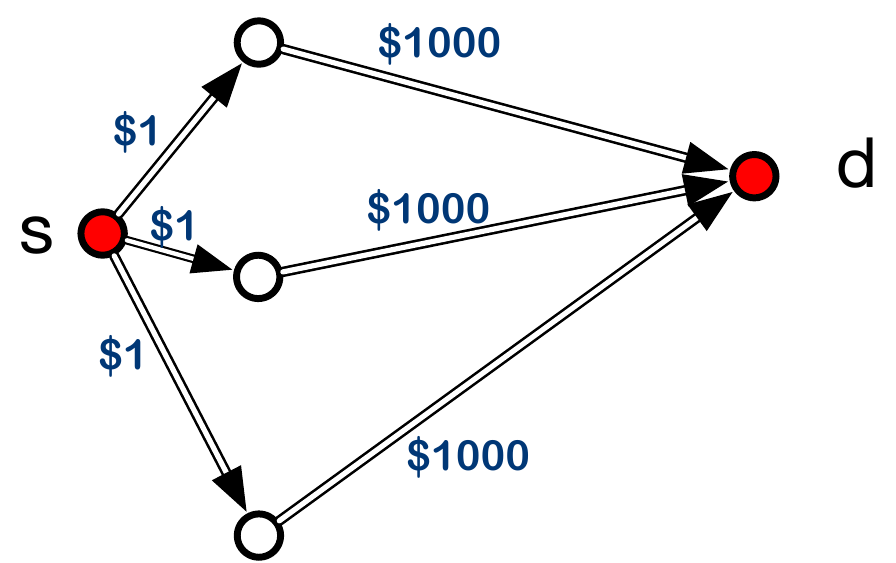}
\caption{A network with different communication costs. Sending an actual message chain from~$s$ to~$d$ costs \$1001. 
\\[-3ex]
\label{fig:ONet}}
\end{minipage}
\hfill
\begin{minipage}[b]{.5\textwidth}
\centering
\includegraphics[width=.8\textwidth]{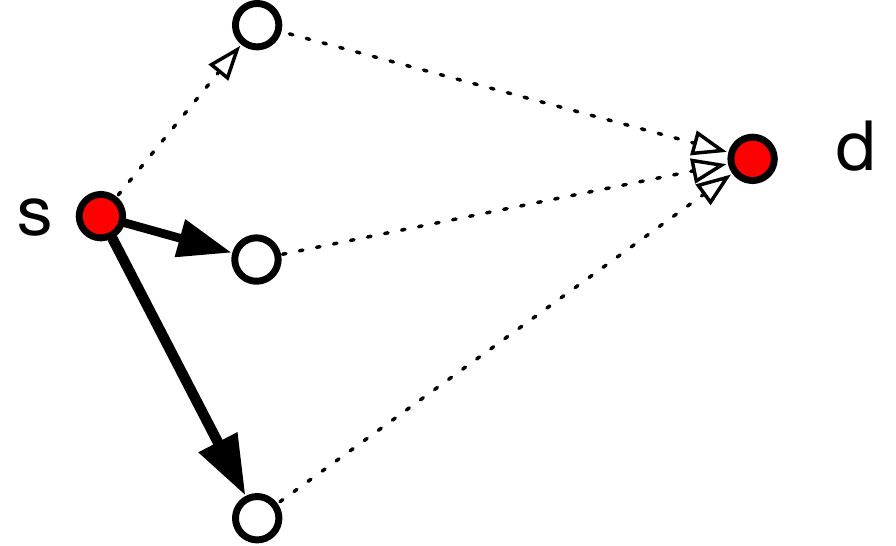}
\caption{The sender~$s$ can inform~$d$ that $v_s=1$ at a cost of~\$2  in a failure free run assuming a bound of $f=2$  failures.  
} \label{fig:ONull} 
\end{minipage}
\end{figure}

This paper investigates the role of null messages for information flow and coordination in synchronous systems with crash failures. Its main contributions are:
\begin{enumerate}
\item We provide a new  definition of  null messages, whereby not sending a message is considered to be a null message only if it conveys nontrivial information. 
Moreover, we  formalize an essential aspect of the synchronous model, by proving that {\em enhanced message chains}, which can contain both actual and null messages, are necessary 
for information transfer in synchronous systems.
\item We strengthen this result by proving   that in order for there to be information transfer from a process~$s$ to~$d$, process~$d$ must {\bf know} that an enhanced message chain from~$s$ has reached it. This result plays an important role in the analysis of information transfer in the presence of failures. 
\item We identify communication patterns called {\em resilient message blocks}, which are necessary for information transfer in crash-prone synchronous systems, proving a stronger and more general theorem than the Silent Choir Theorem of \cite{silence}. 
Based on this theorem, we provide necessary and sufficient conditions that characterize protocols for 
nice-run information transfer, in which the transfer should succeed in failure-free executions, and for Robust information transfer in which it should succeed more generally. 
\item 
Finally, we provide an analysis of communication requirements from protocols solving  the Ordered Response coordination problem, based on resilient blocks. 
\end{enumerate}

This paper is structured as follows. 
In \Cref{sec:model}, we define the model and present preliminary definitions and results regarding null messages, knowledge, and communication graphs, which are used throughout the analysis.  
We show in \Cref{sec:nullmsgsemc} that it is impossible to achieve information transfer from a process $s$ to another process $d$ in synchronous systems without constructing an enhanced message chain from $s$ to $d$. We then prove the stronger fact  (in \Cref{cor:KoP-msg-chain}) that in order for $d$ to know the value of $v_s$, it must {\em know} that an enhanced message chain from $s$ to $d$ exists. 
This is followed in \Cref{sec:failures} by an analysis that identifies  the communication patterns that must be created by a protocol in runs where $d$ learns the value of~$v_s$ \Cref{thm:f-block-nec}. These structures, called {\em resilient message blocks},  can involve multiple enhanced message chains that must be arranged in a particular manner. 
The analysis is applied to characterize necessary and sufficient conditions on the communication patterns solving Nice-run Information Transfer in \Cref{sec:niceIT}. Based on \Cref{sec:failures} and the results of application of \Cref{sec:niceIT},   necessary and sufficient conditions for solving the Ordered Response coordination problem are presented in \Cref{sec:OR}.  Finally, patterns solving Robust Information Transfer are characterized in \Cref{sec:robust-signal}. A brief conclusion is presented in \Cref{sec:Conclusions}.

\subsection{Related Work}
Lamport's seminal paper \cite{Lamclocks} focuses on the role of message chains in asynchronous message passing systems. Indeed, Chandy and Misra showed in \cite{ChM} that the only way in which knowledge about the state of the system at remote sites can be gained in asynchronous systems is via the construction of message chains. As mentioned above, in his later paper \cite{lamp}, Lamport points out that in synchronous systems information can also be conveyed using null messages. 
In a more recent paper \cite{ben2014beyond}, Ben Zvi and Moses analyzed knowledge gain and coordination in a model in which processes are reliable (no process ever crashes) and share a global clock, and  there are upper bounds (possibly greater than~1) on message transmission times along each of the channels in the network. 
They extend the notion of a message chain to so-called {\em syncausal} message chains, which are sequences consisting of a combination of  time intervals that correspond to the upper bounds and actual messages.
They show that syncausal chains are  necessary and sufficient for point-to-point information transfer when $f=0$. Moreover, they define a coordination problem called Ordered Response (which we revisit in \Cref{sec:OR}) and show that a communication pattern they call a {\em centipede}, which generalizes message chains for their model, is necessary and sufficient for solving this problem.
As mentioned in the Introduction and as will be defined in \Cref{sec:model}, not sending a message  does not always count as a null message, even if message delays are bounded.
The notion of an enhanced message chain thus refines that of a syncausal message chain: 
Every  enhanced message chain is  a syncausal message chain, but the converse is not true.

Our paper extends  the work \cite{silence}. Their Silent Choir Theorem, discussed in the previous section, gives necessary conditions that are not sufficient even for failure-free executions. 
In the current paper we take the further step of characterizing necessary and sufficient properties of communication patterns that solve this problem, and investigate the role of silence in the more general coordination problem of Ordered Response coordination problem. 
None of the previous works (e.g., \cite{lamp,ben2014beyond,silence}) requires not sending a message to be informative in order to count as a null message. Making this requirement plays a technically significant role in the analyses performed in the current paper.

In \Cref{sec:niceIT,sec:OR,sec:robust-signal} we consider the design of protocols that are required to behave in a good way in the common case, which in this case is when the initial values are appropriate and no failures occur. 
Focusing on the design of protocols that are optimized for the common case has a long tradition in distributed computing (see, e.g., \cite{liskov1993practical,FastSlow, fairLedger}). In our synchronous model Amdur, Weber, Hadzilacos and Halpern use them in order to design efficient protocols for Byzantine agreement \cite{amdur1992message,hadzilacos1993message}. Guerraoui and Wang and others use them for Atomic Commitment \cite{GW17, silence}. Solutions for Consensus in a synchronous Byzantine model optimized for the common case appeared in \cite{silence2}, also making explicit use of null messages. 
However, how null messages can be used for information transfer and coordination has not been  characterized in a formal way.

\section{Model and Preliminary Results}
\label{sec:model}

We follow the modeling of \cite{silence}. We consider a standard synchronous message-passing model with a set~$\Proc$ of~$N> 2$ processes and benign crash failures. 
For convenience, one of the processes will be denoted by~$s$ and called the {\em source}, while another, $d$ can be considered as the {\em destination}.
Processes are connected via a communication network defined by a directed graph $(\Proc,\chan{})$ where an edge from process~$i$ to process~$j$ is called a {\em channel}, and denoted by $\chan{i,j}$. We assume that the receiver of a message detects the channel over which it was delivered, and thus knows the identity of the sender.
The model is  {\em synchronous}: All processes share a discrete global clock that starts at time~$0$ and  advances by increments of one. 
Communication in the system proceeds in a sequence of \emph{rounds}, with round~$m+1$ taking place between time~$m$ and time~$m+1$, for~$m\ge 0$. A message sent at time~$m$ (i.e., in round~$m+1$) from a process~$i$ to~$j$ will reach~$j$ at time~$m+1$, i.e., at the end of round~$m+1$. In every round, each process performs local computations, sends a set of messages to other processes, and finally receives messages sent to it by other processes during the same round. At any given time~$m\ge 0$, a process is in a well-defined  \defemph{local state}. We denote by $r_i(m)$ the local state of process $i$ at time $m$ in the run $r$. For simplicity, we assume that the local state of a process~$i$  consists of its initial value~$v_i$, the current time~$m$, and the sequence of the events  that~$i$ has observed (including the messages it has sent and received) up to that time. 
In particular, its local state at time~0 has the form~$(v_i,0,\{ \})$. We focus on deterministic protocols, so a  \defemph{protocol} $Q$ describes what messages a process should send and what decisions it should take, as a  function of its local state.

Processes in our model are prone to crash failures. A faulty process in a given execution fails by \emph{crashing} at a given time. A process that crashes at time~$t$ is completely inactive from time $t+1$ on, and so it  performs no actions and in particular sends no messages in  round~$t+2$ and in all later rounds. It behaves correctly up to and including round~$t$. Finally, in round $t+1$ (which takes place between time~$t$ and time~$t+1$) this process sends a (possibly strict) subset of the messages prescribed by the protocol. For ease of exposition we assume that the process does not perform any additional local actions (e.g., decisions) in round~$t+1$.

For ease of exposition, we say that a process that has not failed up to and including time~$t$ is {\em active} at time $t$. We will consider the design of protocols that are required to tolerate up to~$f$ crashes.
We denote by~$\modelf$ the model described above in which  no more than~$f$ processes crash in any given run. 
We assume that a protocol has access to the values of~$N$ and~$f$ as well as to the communication  network $(\Proc,\chan{})$.
A \defemph{run} is a description of a (possibly infinite) execution of the system.
We call a set of runs a \defemph{system}.
We will be interested in systems of the form~$R_Q=R(Q,\modelf )$ consisting of all runs of a given protocol~$Q$ in which no more than~$f$ processes fail. 
A failure pattern determines who fails in the run, and what messages it succeeds in sending when it fails. Formally, we define: 
\begin{definition}[Failure patterns]\label{def:fp}
A {\em failure pattern} for a model $\gamma^f$ is a set \[FP\triangleq\{\langle q_1,t_{1},Bl(q_1)\rangle,\ldots \langle q_k,t_{k},Bl(q_k)\rangle\}\] of $k\le f$ triples, where $q_i\in\Proc$,  $t_{i}\ge 0$, and  $Bl(q_i)\subseteq \Proc$ for $i=1,\ldots, k$.
We consider a run~$r$ to be {\em compatible with} a failure pattern  $FP=\{\langle q_1,t_{1},Bl(q_1)\rangle,\ldots \langle q_k,t_{k},Bl(q_k)\rangle\}$ if 
 
\begin{enumerate}
    \item each process $q_i$ fails in $r$ at time $t_{i}$ and only the processes $q_1,\ldots,q_k$ fail in $r$, while 
    \item for every process~$p$ to whom $q_i$ should send a message in round~$t_{i}+1$  according to~$Q$ (based on $q_i$'s local state  at time~$t_{i}$ in~$r$), a message  from~$q_i$ to~$p$ is sent at time $t_i$ in~$r$ iff $p\notin Bl(q_i)$. 
\end{enumerate}
\end{definition}
If a process $q_i$ is specified as failing in a pattern~$FP$, then for every  $p\in Bl(q_i)$, we consider the channel $\chan{q_i,p}$ to be {\em blocked} from  round $t_i+1$ on.

For ease of exposition in this paper we will restrict our attention to the case in which the source~$s$ has a binary initial value $v_s\in\{0,1\}$, while the initial values of all other processes $p\ne s$ are fixed. Thus, there are only two distinct initial global states in a system $R_Q$. Moreover,  the deterministic protocol~$Q$, a given initial global state (in our case the value of $v_s$) and a failure pattern uniquely determine a run. Relaxing these assumptions would not modify our results in a significant way; it would only make proofs quite a bit more cumbersome. 

\subsection{Defining Knowledge}
Our analysis makes use of a formal theory of knowledge in distributed systems.  
We sketch the theory here; see \cite{fhmv} for more details and a general introduction to the topic. 
In general, a process~$i$ can be in the same local state in different runs of the same protocol. We shall say that two runs~$r$ and~$r'$ are {\em indistinguishable} to process~$i$ at time~$m$ if~$r_i(m)=r'_i(m)$.  
The current time~$m$  is represented in the local state~$r_i(m)$, and so, $r_i(m)=r'_i(m')$ can hold only if~$m=m'$.  Notice that since we assume that processes follow deterministic protocols, if $r_i(m)=r'_i(m)$ then process~$i$ is guaranteed to perform the same actions at time~$m$ in both~$r$ and~$r'$ if it is active at time~$m$.

\begin{definition}[Knowledge]
\label{def:know}
	Fix a system~$R$, a run~$r\in R$, a process~$i$ and a fact~$\varphi$.  
	We say that~$K_i\varphi$ (which we read as ``process~$i$ \defemph{knows}~$\varphi$'') holds at time~$m$ in~$r$ iff~$\varphi$ is true at time~$m$ at
 all runs~$r'\in R$ such that 
	$r_i(m)=r'_i(m)$. 
\end{definition}

\Cref{def:know} immediately implies the so-called {\em Knowledge property}: If~$K_i\varphi$ holds at time~$m$ in~$r$, then so does~$\varphi$. 
The logical notation for ``the fact~$\varphi$ holds at time~$m$ in the run~$r$ with respect to the system~$R$'' is $(R,r,m)\sat \varphi$. 
 Often, the system is clear from context and is not stated explicitly. In this paper, the system will typically consist of all the runs of a given protocol~$Q$ in the current model of computation, which we denote by $R_Q$. 
Observe that knowledge can change over time. Thus,  for example,~$K_j(v_i=1)$ may be false at time~$m$ in a run~$r$ and true at time~$m+1$, based perhaps on  messages that~$j$ does or does not receive in round~$m+1$.

An essential connection between knowledge and action in distributed protocols, called the \defemph{knowledge of preconditions principle} (KoP), is provided in \cite{Moses16}. It states that whatever must be true whenever a particular action is performed by a process~$i$ must be known by~$i$ when the action is performed.
More formally, we say that a fact~$\varphi$ is a \defemph{necessary condition} for an action~$\act$ in a system~$R$ if for all runs~$r\in R$ and times~$m$, if~$\act$ is performed at time~$m$ in~$r$ 
then~$\varphi$ must be true at time~$m$ in~$r$. In our model the KoP 
can be stated as follows: 
\begin{theorem}[KoP \cite{Moses16}]
	\label{thm:kop}
	Fix a protocol~$Q$ for~$\modelf$ and let~$\act$ be an action of process~$i$ in~$\RP$. 
	If~$\varphi$ is a necessary condition for~$\act$ in~$\RP$ then 
	$K_i\varphi$ is a necessary condition for~$\act$ in~$\RP$.
\end{theorem}

As observed in \cite{BzM2011}, in synchronous systems the passage of time can provide a process information about events at remote sites. (E.g., $p$ can know that~$q$ performs an action at some specific time, based purely on the protocol.) 
In order to focus on genuine flow of information between processes, we make the following definition:

\begin{definition}
    {\em Information transfer} (\IT) between~$s$ and~$d$ is achieved when $K_d(v_s=b)$ holds, for some value $b\in\{0,1\}$. 
\end{definition}

Since the initial value $v_s$ is independent of the protocol, for~$d$ to learn this value requires genuine flow of information from~$s$ to~$d$. 

\subsection{Null Messages and Enhanced Message Chains}\label{sec:nullmsgsemc}

As discussed in the Introduction, if no message is ever sent over a given channel at time~$t$ under the protocol~$Q$, then the absence of such a message in a given execution is not informative. We now define not sending to be a null message only if it is informative: 

\begin{definition}
\label{def:null}
Let $r$ be a run of some protocol $Q$. Process $\pmb{i}$ {\bf  sends $\pmb{j}$ a  null message} 
at~$(r,t)$ if
   \begin{itemize}
       \item  $\chan{i,j}$ is not blocked at $(r,t)$,
       \item $i$ does not send an actual message over $\chan{i,j}$ at $(r,t)$, and
       \item there is a run $r'\!$ of $Q$ in which $i$ sends an  actual  message over $\chan{i,j}$ at $(r',t)$.
   \end{itemize}
\end{definition}

We can now generalize message chains to allow for null messages as well as actual ones: 

\begin{notation}
    We denote by $\theta=\node{p,t}$ the {\bf process-time} pair consisting of a process $p$ and time $t$.
    Such a pair is used to refer to the point at time~$t$ on~$p$'s timeline.
\end{notation}

\begin{definition}
\label{def:simple-msg-chain}
    Let $r$ be a run of a protocol $Q$. 
    We say that there is {\bf an enhanced message chain} from~${\theta}=\node{p,t}$ to 
 ${\theta'}=\node{q,t'}$ in~$r$, and write $\pmb{\theta\emc_{Q,r}\theta'}$, if there exist processes $p=i_1,i_2\ldots, i_k=q$ and times $t\le t_1<t_2<\cdots< t_{k}= t'$ such that for all $1\le h<k$ process~$i_h$ sends either an actual message or a null message to~$i_{h+1}$ at $(r,t_h)$. (We omit the subscript and write simply $\pmb{\theta\emc\theta'}$ when~$Q$ and~$r$ are clear from the context.)
     
\end{definition}
Observe that \Cref{fig:ONull} contains three enhanced message chains between process~$s$ and~$d$. Two of them contain a single actual message each, and one does not contain any actual message. 

We are now ready to show  that information transfer in synchronous systems requires the existence of an enhanced message chain. 

\begin{theorem}\label{thm:msg-chain}
Let $f\ge 0$ and let~$Q$ be a protocol and  $r\in R(Q,\modelf)$. Then  $K_d(v_s=1)$ holds  at $(r,m)$ only if $\node{s,0}\rightsquigarrow\node{d,m}$ in $r$.
\end{theorem}
\begin{proof}
Assume, by way of contradiction, that $K_d(v_s=1)$ at $(r,m)$, and that  $\node{s,0}\not\rightsquigarrow \node{d,m}$ in $r$.  
By \Cref{def:know} it suffices to show a run~$r'\in R(Q,\modelf)$ in which $v_s\ne 1$ such that $r_d(m)=r'_d(m)$. 
Denote 
\[T~\triangleq~ \{\theta\in\mathbb{V}\!:\,\node{s,0}\not\rightsquigarrow\theta \text{ in } r\}.\]
I.e., $T$ is the set of nodes to which there is no enhanced message chain from $\node{s,0}$ in the run~$r$. Observe that, by assumption, $\node{d,m}\in T$. We construct a run~$r'$ as follows: 
The initial global state $r'(0)$ differs from $r(0)$ only in the value of the variable $v_s$ (thus, $v_s=0$ in~$r'$), which appears in $s$'s local state. All other initial local states are the same in $r'(0)$ and in $r(0)$. Finally, all processes have the same failure patterns in both runs. 
We now prove by induction on time~$t$ that $r_{i}(t)=r'_{i}(t)$ holds for all nodes $\node{i,t}\in T$.

\uline{Base}:~$t=0$. Assume that $\node{i,0}\in T$. By definition of~$T$, it follows that $i\ne s$, and by construction of $r'$ we immediately have that $r_{i}(0)=r'_{i}(0)$, as required. 

\uline{Step}: Let~$t>0$ and assume that the claim holds for all nodes $\node{j,t'}$ with $t'<t$.
Fix a node~$\node{i,t}\in T$.
Clearly,  $\node{i,t-1}\in T$, and so by the inductive hypothesis $r_{i}(t-1)=r'_{i}(t-1)$. 
To establish our claim regarding $\node{i,t}$, it suffices to show that~$i$ receives exactly the same messages at time~$t$ in both runs. 
Recall that  the synchrony of the model implies that the only messages that~$i$ can receive at time~$t$  are  ones sent at time~$t-1$. Hence, we reason by cases, showing that  every process~$z\ne i$ sends~$i$ the same messages at time~$t-1$ in both runs. 
\begin{itemize}
\item Suppose that $\node{z,t-1}\in T$. 
    Then by the inductive assumption $r_z(t-1)=r'_z(t-1)$, i.e., process $z$ has the same local state at time~$t-1$  in both runs. Since~$Q$ is deterministic and since the runs $r$ and $r'$ have identical failure patterns, $z$ sends~$i$ a message in $r$ at time~$t-1$ in~$r'$ iff it does so in~$r$. Moreover, if it sends a message, it sends the same message in both cases.
\item Suppose that $\node{z,t-1}\notin T$, i.e., there is an enhanced message chain from $\node{s,0}$ to $\node{z,t-1}$ in $r$. Since $\node{i,t}\in T$ we have that $z$ does not send a message to $i$ at time $t-1$ in $r$. Assume by way of contradiction that $i$ receives such a message in~$r'$. In particular, this implies that the channel $\chan{z,i}$ is not blocked in $r'$, and since $r$ and $r'$ have the same failure pattern, $\chan{z,i}$ is not blocked in $r$.  
Hence, by definition, there is a null message from $\node{z,t-1}$ to $\node{i,t}$ in $r$. This contradicts the fact that, by assumption, $\node{i,t}\in T$.
It follows that, in both~$r$ and~$r'$, process $i$ does not receive any message from~$z$ at time~$t$.
\end{itemize}
Since $r_{i}(t-1)=r'_{i}(t-1)$ process~$i$ performs the same actions at time~$t-1$ in both runs. Since, in addition, $i$ receives exactly the same messages at $(r',t)$ as it does in~$(r,t)$ as we have shown, it follows that $r_i(t)=r'_i(t)$. 

The inductive argument above showed that, for all processes~$i$ and all times~$t\le m$, if~$i$ is has not failed by time~$t$ and  $\node{i,t}\in T$, then  $r_i(t)=r'_i(t)$.
Since $\node{d,m}\in T$ by assumption, it follows that, in particular, 
$r_{d}(m)=r'_{d}(m)$. 
Since $v_s \ne  1$ in~$r'$ we obtain that~$\neg K_{d}(v_s =1)$ at time~$m$ in~$r$ by \Cref{def:know}. This contradicts the assumption that $K_{d}(v_s =1)$ holds at time~$m$ in~$r$, completing the proof.
\end{proof}

\Cref{thm:msg-chain} establishes that enhanced message chains are necessary for information transfer for all values of $f\geq 0$. In fact, when $f=0$, enhanced message chains are also sufficient.  (We omit a proof of this particular claim since it will follow from the more general \Cref{thm:OSN tight-suff}). 
This demonstrates that enhanced message chains play an analogous role in reliable synchronous settings to the one that standard message chains play in asynchronous systems (cf.~\cite{ChM}). 
 
In reliable systems, null messages are detected as reliably as actual messages are.
As discussed  in the Introduction, however, this is no longer true in the presence of failures. 
The knowledge formalism allows us to crisply capture a stronger requirement than the one in \Cref{thm:msg-chain}, which lies at the heart of the issue.
 
\begin{theorem}\label{cor:KoP-msg-chain}
Let~$f\ge 0$, let  $r$ be a run of $R_Q=R(Q,\modelf)$ and denote $\theta_s=\node{s,0}$ and  $\theta_d=\node{d,m}$. Then $K_d(v_s=1)$ holds  at $(r,m)$ only if  $\;K_d(\theta_s\rightsquigarrow\theta_d)$ holds at $(r,m)$.
\end{theorem}

\begin{proof}
Suppose that $(R_Q,r,m)\sat K_d(v_s=1)$. \Cref{def:know} implies that $(R_Q,r',m)\sat K_d(v_s=1)$ holds for every run~$r'$ such that $r'_d(m)=r_d(m)$. By \Cref{thm:msg-chain} it follows that 
$(R_Q,r',m)\sat(\theta_s\rightsquigarrow\theta_d)$ for every such~$r'$, and so by \Cref{def:know} we obtain that $(R_Q,r,m)\sat K_d(\theta_s\rightsquigarrow\theta_d)$, as claimed. 
\end{proof}

\section{Dealing with Failures}
\label{sec:failures}
The need to know that an enhanced chain has reached~$d$, established in \Cref{cor:KoP-msg-chain}, is not the same as the mere existence of such a chain. 
This difference matters when processes may fail, because then silence can be ambiguous, and null messages can be confused with process crashes. 
How, then, can~$d$ come to know that an enhanced chain has reached it, in a setting where $f>0$ processes can crash? 
One possibility would be to have the protocol construct $f+1$ enhanced chains from $\node{s,0}$ to $\node{d,m}$ whose sets of participating processes are pairwise disjoint.%
\footnote{A similar issue arises in the  Byzantine Agreement literature (cf.~\cite{dolev1982byzantine})
where many process-disjoint chains  are used to overcome the possibility of failures.}
While such an assumption may be needed in a protocol that in a precise sense  guarantees information transfer (we revisit this point in \Cref{sec:robust-signal}), in many instances the destination process can learn the sender's value even if the protocol does not employ such a scheme. 

Our purpose is to investigate the communication patterns under which~$d$ can learn the  value of~$v_s$. 
 For this purpose, we will find it convenient to associate a ``communication graph'' with every run, which we define as follows.  The nodes of the graph are process-time pairs. Edges correspond to messages sent among processes, to~null messages, and a local tick of the clock at a process. More formally:

\begin{definition}[Communication Graphs]
\label{def:CGqr}
The {\em communication graph} of a run $r$ of protocol $Q$ is  $\CG_Q(r)\triangleq (\mathbb{V},E)$, with nodes $\mathbb{V}=\Proc\times \mathbb{N}$ and edges $E=\Eloc\cup \Eactual(r)\cup \Enull(r)$, where 
\begin{itemize}
    \item  $\Eloc=\left\{ (\node{i,t},\node{ i,t+1}):i\in\Proc, t\in\mathbb{N}\right\}$,
    \item $\Eactual(r)$~$=\left\{ (\node{ i,t},\node{ j,t+1}):\text{${i}$ sends an actual message to ${j}$ at time ${t}$ in ${r}$}\right\}$,
    \item $\Enull(r)=\{ (\node{ i,t},\node{j,t+1}):i$ sends a null message to $j$ at time $t$ in~$r\}$
\end{itemize}
\end{definition}

Notice that both the set of nodes~$\mathbb{V}$ and the set $\Eloc$ of local edges are the same in all  communication graphs. 
Observe that the communication graph directly represents enhanced message chains: 
$\pmb{\theta\rightsquigarrow_{Q,r}\theta'}$ holds if and only  if $\CG_Q(r)$ contains a path from~$\theta$ to~$\theta'$. 

\subsection{Resilient Message Blocks}
\label{sec:blocks}

The Silent Choir Theorem of \cite{silence} states that a necessary condition for $K_d(v_s=1)$ to hold at time~$m$ without an actual chain from~$s$ to~$d$, is for there to be actual message chains to $f+1$ members of the silent choir, after which they are all silent to~$d$ at time~$m-1$. As discussed in the Introduction, however, these members need not send~$d$ a {\em null message} at~$m-1$. Indeed, members of the ``choir'' need not even be neighbors of~$d$. In this section we will present a strictly stronger condition on the communication pattern than the one in their theorem, called a resilient message block. The new condition will also be more informative as it is explicitly formulated in terms of  null messages.  
Moreover, for the interesting case of optimizing communication for  failure-free executions, our resilient message blocks will be both necessary {\em and} sufficient for information transfer.

Very roughly speaking, a process that knows about failures might detect the existence of an enhanced chain more easily than one who is unaware of failures. E.g., if $d$ has detected all~$f$ faulty processes, then it can readily detect null messages sent by correct processes. Correctly coping with such issues requires a somewhat subtle definition and theorem statement.

\begin{notation}[$B$ null free paths]
Fix a protocol~$Q$, let~$r$ be a run of~$Q$, and let~$B$ be a set of processes. 
A path~$\pi$ in $\CG_Q(r)$ that does not contain null messages sent by processes in the set~$B$  is called {\em $B$ null free} (we write that ``$\pi$ is  $\Bnullfree$'' for brevity). 
\end{notation}

A $\Bnullfree$ path can contain null messages, but not ones ``sent'' by members of~$B$. Roughly speaking, if the members of~$B$ crash, this path can remain a legal enhanced message chain. In light of \Cref{cor:KoP-msg-chain}, we are now ready to characterize the properties of communication graphs of protocols that enable information transfer. 
Using $\Bnullfree$ paths we can now define the central communication patterns that will play a role in our analysis: 
\begin{definition}[$(\fF)$-resilient message block]\label{def:msg-res-block-gen}
Let $r$ be a run of a protocol $Q$ and denote by $\faulty^r$ the set of processes that fail in $r$. Let $\theta,\theta'\in\Proc\times\Nat$ be two nodes. An {\em $(\fF)$-resilient message block from $\theta$ to $\theta'$ in} $\CG_Q(r)$ is a set~$\Gamma$ of paths between~$\theta$ and~$\theta'$ such that for every set of processes $B$ such that $|B\cup \faulty^r| \leq f$,  
there is a $\Bnullfree$ path in $\CG_Q(r)$ from $\theta$ to $\theta'$ in~$\Gamma$.
\end{definition}

Recall that an actual message chain contains no null messages, and is thus a $\Bnullfree$ path for every set $B$ of processes. As a result, an actual message chain between two nodes is, in particular, an $\fF$-resilient message block, for all runs and {\em all} values of~$f\ge0$. 
Our next theorem states that in order for a process~$d$ to know in some run $r$ at time $m$ that there is an enhanced message chain from some node to itself, there must be a resilient message block between them.
Roughly speaking, the claim is proved by way of contradiction. We assume a set $B$ of processes contradicting the assumption and construct a run~$r'$ that is indistinguishable to~$d$ from~$r$ in which nodes involving processes of $B$ cut all paths from $\theta_s=\node{s,0}$ to  $\theta_d=\node{d,m}$. I.e., there is no enhanced message chain from $\theta_s$ to $\theta_d$ in $r'$. More precisely, given a set $B$ of processes we will define the set $T_B$ to be the set of nodes to which there is {\em no} $\Bnullfree$ path from $\theta_s$. We give an example of such a set in \Cref{fig:proof}. The highlighted nodes are in $T_B$ while the others are not in $T_B$.\footnote{For the sake of clarity we do not draw all of the nodes and edges of the communication graph. Thus, for example, we do not represent all the nodes along local time lines and the local edges in $\Eloc$ that connect them.} 
As detailed in the complete proof, when constructing $r'$, we make the processes of $B$ fail at times that make these failures unnoticeable by  processes appearing in $T_B$ (and hence by the contradiction assumption neither by $d$ at time $m$). As a result, process $d$ does not know at $(r,m)$ that an enhanced message chain has reached it. Since by  \Cref{thm:msg-chain}, this is a necessary condition, we conclude that  $\neg K_d(v_s=1)$ at $(r,m)$, as claimed.  We can now show:

\begin{figure}[!h]
\centering
    \includegraphics[width=11 cm]{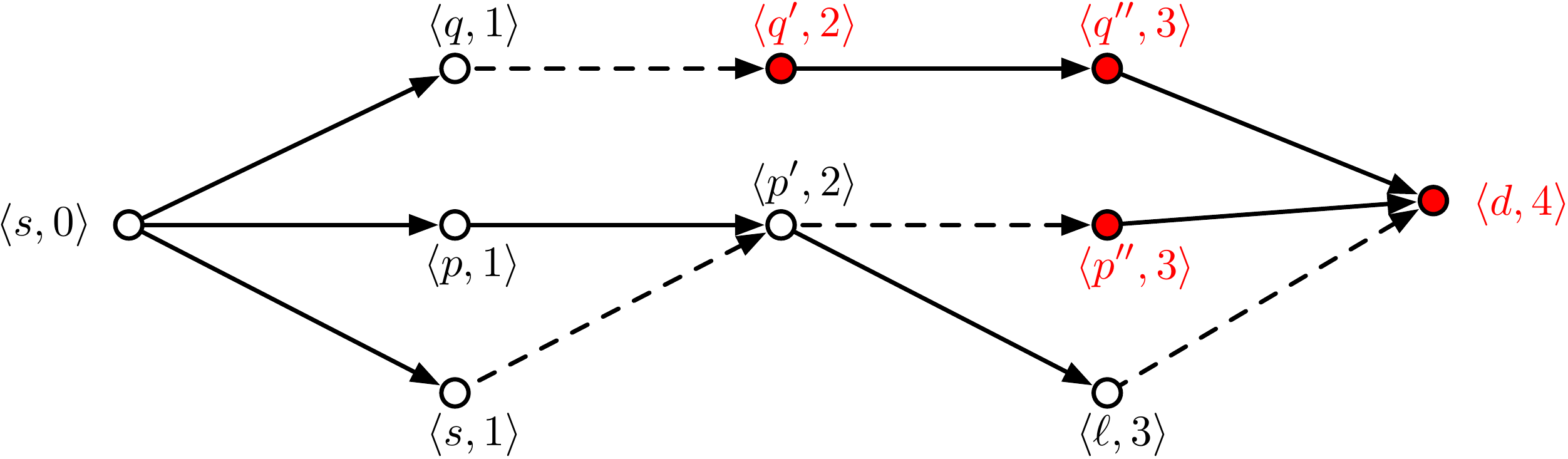}
    \caption{A communication graph and its corresponding set $T_B$ (highlighted) for $B=\{q,p',\ell\}$.}
    \label{fig:proof}
\end{figure}


\begin{theorem}\label{thm:f-block-nec}
Let $r$ be a run of a given protocol $Q$ and let $\faulty^r$ denote the set of faulty processes in $r$.  If~$K_d(v_s=1)$ holds at $(r,m)$, there is an $(\fF)$-resilient message block from $\theta_s\triangleq\node{s,0}$ to $\theta_d\triangleq\node{d,m}$ in $\CG_Q(r)$.
\end{theorem}

The complete proof appears in the Appendix.%
\footnote{Every claim is completely proved either in the main  text or in the Appendix.}
We point out that \Cref{thm:f-block-nec} is stronger than the Silent Choir Theorem of \cite{silence}. Namely, as we claim in \Cref{lem:strongerthansilentchoir} the existence of the depicted resilient message block implies that a silent choir exists. However, the converse is not true. Remember the example of \Cref{figure1}. The set of processes $\{s,p\}$ is a silent choir for $f=1$ but clearly, there is no $f$-resilient message block --- take for instance $B=\{p\}$, there is no $\Bnullfree$ path from $\node{s,0}$ to $\node{d,2}$.

As the following lemma establishes, \Cref{thm:f-block-nec} implies the  Silent Choir Theorem: 
\begin{lemma}\label{lem:strongerthansilentchoir}
If there is an $\fF$-resilient message block from $\theta$ to $\theta'$ in $CG_Q(r)$, then there exists a silent choir from $\theta$ to $\theta'$ in $r$.
\end{lemma}

We are now ready to characterize the communication structures needed to solve information transfer and coordination in several interesting cases.  
\section{Application: Information Transfer in Nice Runs}
\label{sec:niceIT}
\Cref{thm:f-block-nec} is at the heart of information transfer in fault-prone synchronous systems. \mbox{$f$-resilient} message blocks constitute a necessary pattern for information transfer in our model since without them, one cannot know that an enhanced message chain reaches it.
Clearly, in a system prone to crash failures, it is natural to require for information to be conveyed in failure-free runs. Focusing on the design of protocols that are optimized for the common case has a long tradition in distributed computing and can be useful in applications including Consensus, Atomic Commitment, and blockchain protocols (see, e.g., \cite{amdur1992message,hadzilacos1993message,liskov1993practical,FastSlow, fairLedger, silence2}).

Clearly, the information that a null message conveys depends crucially on the protocol. More precisely, it depends on the conditions under which $i$ would not send an {\em actual} message. To account for this, we make the following definition: 

\begin{definition}
\label{def:null-phi}
 We say that    process~$i$ {\bf sends  a null message} over $\chan{i,j}$ at time~$t$ {\bf in case}~$\pmb{\varphi}$ {\bf in} a given protocol~$\pmb{Q}$ if for every run~$r\in R_Q$  in which $\chan{i,j}$ is not blocked at $(r,t)$, process~$i$ sends a null message over $\chan{i,j}$ at $(r,t)$ ~iff~ $(R_Q,r,t)\sat\varphi$.
\end{definition}
 Roughly speaking, a process~$p$ ``sends a null message in case $\varphi$'' at time~$t$ if whenever it is active at time~$t$  and $\varphi$ does not hold, then~$p$ sends an actual message. If~$\varphi$ does hold at time~$t$, then~$p$ keeps silent.
By definition, a null message is sent in case~$\varphi$ only if $\varphi$ is true. Therefore, in a reliable system (i.e., when $f=0$), sending a null message in case~$\varphi$ informs the recipient that~$\varphi$ holds.

We focus on solving the \IT\ problem in {\em nice runs}, which are formally defined as follows:
\begin{definition}
Let $Q$ be a protocol. The run of $Q$ in which $v_s=1$ and no process fails is called $\pmb{Q}$'s {\bf  nice run}. We denote $Q$'s nice run by $\pmb{\rnice(Q)}$ and the communication graph of $\rnice(Q)$ by $\pmb{\nG(Q)}$. When $Q$ is clear from the context, we simply write $\pmb{\rnice}$.
\end{definition}
Every protocol~$Q$ has a unique nice run.  We show in this section that the conditions of \Cref{thm:f-block-nec} are also sufficient for information transfer in {\em nice} runs.
Clearly, in a failure-free execution~$r$, it holds that $\faulty^r=\emptyset$. An $\fF$-resilient message block with $\faulty^r=\emptyset$ is a central structure for information transfer in nice runs and will be used in \Cref{sec:OR}. It can be defined in slightly simpler terms as follows:
\begin{definition}[$f$-resilient message block]\label{def:msg-res-block}Let $\theta,\theta'\in\Proc\times\Nat$ be two nodes. An {\em $f$-resilient message block from $\theta$ to $\theta'$ in} $\CG_Q(r)$ is a set~$\Gamma$ of paths between~$\theta$ and~$\theta'$ such that for every set of processes $B$ of size $|B|\leq f$,  there is a $\Bnullfree$ path in $\CG_Q(r)$ from $\theta$ to $\theta'$ in~$\Gamma$. 
\end{definition}

Observe that the presence of an $f$-resilient message block in the communication graph of a run $r$ suffices to ensure that in any run $r'$ that ``looks the same'' to~$d$ at $(r',m)$ i.e., $r_d(m)=r'_d(m)$, there is an enhanced message chain reaching $d$.

\begin{theorem}[Nice-run \IT]\label{thm:OSN tight}  
$f$-resilient message blocks are necessary and sufficient for solving~\IT\ in the nice run. Namely, 
\begin{itemize}
    \item (Necessity) If $K_d(v_s=1)$ at $(\rnice,m)$ then there is an $f$-resilient message block from $\theta_s$ to $\theta_d$ in $\nG$.
    \item \label{thm:OSN tight-suff} (Sufficiency) If a communication graph $\CG$ contains an $f$-resilient message block between $\theta_s$ and $\theta_d$, then there exists a protocol $Q$ such that $\nG(Q)=\CG$ that solves \IT\ between $\theta_s$ and $\theta_d$. 
\end{itemize}
\end{theorem}
\Cref{thm:OSN tight} gives a {\em tight} characterization of the communication patterns needed to solve~\IT\ in nice runs: Every solution must construct an $f$-resilient message block, and for every $f$-resilient message block, there exists an \IT\ protocol that uses only the paths in this block in its nice run. 
The necessity part of \Cref{thm:OSN tight} results from \Cref{thm:f-block-nec} applied to $\rnice$ where the set of faulty processes is $\faulty^r=\emptyset$. To show sufficiency, we need to describe a protocol~$Q$ as claimed. 
Before sketching the proof, we discuss and define a class of protocols used in the proof. 


In protocols solving \IT\ in the nice run, messages only need to convey whether  the sender has detected that the run is not nice. To consider this more formally, for a given protocol~$Q$ we denote by~$\nicerun$ the fact ``the current run is $\rnice\,$.''
Typically, if $f>0$, it is impossible 
for a process to know that $\nicerun$ is true: Even if a process has observed no failures, or indeed, even if no failures have occurred by a given time~$t$, there may be  run indistinguishable to the process in which one or more processes fail after time~$t$. 
Nevertheless,  a process may readily know that $\neg\nicerun$, if it knows of a failure or detects that $v_s=0$. 
Of course, because of the Knowledge property, in the nice run~$\rnice$ itself, no process will ever know $\neg\nicerun$. 

\begin{definition}[Nice-based Message Protocols]
    A protocol $Q$ is a Nice-based Message protocol (\NbM\ protocol) if (i) All actual messages sent are single bit messages and whenever a process~$p$ sends an actual message, it sends a~`0' if $K_{p}\neg\nicerun$ and sends a~`1' otherwise (i.e., if $\neg K_{p}\neg\nicerun$) and (ii) for all processes~$p$, each null message sent by $p$ over any channel is a null message in case~$\neg K_{p}\neg \nicerun$.
\end{definition}
We remark that a process~$p$ can efficiently check  whether  $K_p(\neg\nicerun)$  by simply comparing $p$'s local state at $(r,t)$ to its local state at $(\rnice, t)$. If the two states are identical, then the predicate $K_p(\neg\nicerun)$ is false. Otherwise, it is true.

\begin{proofsketch} (of sufficiency in \Cref{thm:OSN tight}):

Suppose that $\CG$ contains an $f$-resilient message block between $\theta_s$ and $\theta_d$. The desired protocol~$Q$ is defined to be a Nice-based Message protocol such that $\nG(Q)=\CG$. I.e., for every edge $e\triangleq(\node{u,t},\node{v,t+1})$ in $\CG$: If $e\in\Enull$, then $u$ keeps silent if $\neg K(\neg\nicerun)$ and should send a `0' to $v$ otherwise. If $e\in\Eactual$, then $u$ should send `1' to $v$ if $\neg K_u(\neg\nicerun)$ and `0' otherwise. 
We show that the assumptions guarantee that for every run $r$ of $Q$ there is (at least one) path in $\nG(Q)$ from $\theta_s$ to $\theta_d$ along which no silent process fails in $r$. We show by induction on time that in every run $r$ in which $v_s=0$, for each node $\node{p,t}$ along this path, $p$ knows at time $(r,t)$ that the run is not nice. Hence, $d$ also knows at $(r,m)$ that the run is not nice. Since $\nicerun$ holds throughout $\rnice$, so does $\neg K_d \neg \nicerun$. It follows that $K_d(v_s=1)$ at $(\rnice,m)$, as claimed.
\end{proofsketch}

\section{Application: Coordination }\label{sec:OR}
$f$-resilient message blocks and Nice-run \IT\ are useful tools for solving more complex problems. We now show how they can be used to characterize solutions to the Ordered Response (\OR) coordination problem. This problem was originally defined in \cite{ben2014beyond}, and it requires a sequence of actions to be performed in linear temporal order, in response to a triggering signal from the environment.\footnote{In~\cite{ben2014beyond} Ordered Response was studied in a reliable setting with no crashes, and upper bounds on message delivery times; a very  different set of assumptions than here.}  For simplicity, we identify the signal to be received iff $v_s=1$.
We assume that each process $i_h\in\{i_1,i_2,\ldots,i_k\}$  has a specific action~$a_h$ to perform, and that the actions should be performed in order, provided that initially $v_s=1$. 
\begin{definition}[Ordered Response] \label{def:OR-input}
We say that a protocol~$Q$ is \textbf{consistent} with the instance~$\ORi=\angles{\factOR,a_1,\ldots,a_k}$ of the {\em Ordered Response (\OR)} problem if it guarantees that~$a_h$ is performed in a run only if~$v_s=1$ and~$a_1,...,a_{h-1}$ are performed.
In particular, if both~$a_{h-1}$ and~$a_h$ are performed at times~$t_{h-1}$ and~$t_h$ respectively, then~$t_{h-1}\leq t_h$.
Protocol~$Q$ \textbf{solves} this instance \ORi\ if, in addition, all of the actions~$a_h$ are performed in~$Q$'s nice run. 
\end{definition}



Let us denote by $\ba_h$ the fact that the action $a_h$ has (already) been performed. 
Since, by definition of \OR,  both $v_s=1$ and $\ba_{h-1}$ are necessary conditions for performing~$a_h$,  the Knowledge of Preconditions principle (\Cref{thm:kop}) implies that these facts must be known when $\ba_h$ is performed:

\begin{lemma}\label{lem:OR-basic}
	Suppose that~$Q$ solves the instance $\ORi=\angles{\factOR,a_1,\ldots,a_k}$ of \OR. For every run~$r$  of~$Q$ and action $a_h$ performed (at time~$t_h$) in~$r$, we have 
	\begin{enumerate}
	    \item $(R_Q,r,t_h)\sat K_{i_h}(\factOR)$,  and 
	    \item $ (R_Q,r,t_h)\sat K_{i_h}(\;\!\ba_{h-1})$ if $h>1$. 
	\end{enumerate}
\end{lemma}

For a protocol~$Q$ solving an instance of Ordered Response, every action $a_h$,   is performed at some specific time $t_h$ in the nice run $\rnice=\rnice(Q)$. For ease of exposition we denote by by~$\theta_h\triangleq \node{i_h,t_h}$ the  node of $\nG(Q)$ where the action is taken, and by   $\theta_h^+\triangleq \node{i_h,t_h+1}$ the node of $i_h$ one time step after the node $\theta_{h}$.

We can use \Cref{lem:OR-basic} to provide necessary conditions on the nice communication graph of protocols that solve Ordered Response. \Cref{lem:OR-basic}(1) implies that~$Q$ must perform Nice-run \IT\ to $\theta_h$, for all actions $a_h$. 
\Cref{lem:OR-basic}(2), in turn, implies that $i_{h}$ needs to learn that $a_{h-1}$ has been performed in order to perform its action. A straightforward way to do this is by performing Nice-run \IT\  directly between consecutive actions, i.e., by creating an $f$-resilient block between $\theta_{h-1}$ and~$\theta_h$. 
While this is a possible solution, it is {\em not} the only way that $i_h$ can obtain this knowledge. Process~$i_h$ can also learn about the previous action {\em indirectly}. This requires $i_h$ to know that $i_{h-1}$ couldn't have failed (since otherwise~$i_h$ would have an indication that it failed) and that it received the information needed to perform~$a_{h-1}$.
The possibility of acting on indirect information is 
where solving the Ordered Response problem goes beyond Nice-run \IT. We can show the following:

\begin{theorem}[\OR\ Necessity] \label{thm:OR-necc-1}
Let~$Q$ be a protocol
solving $\ORi=\angles{\factOR,a_1,\ldots,a_k}$ for some sequence of times $t_1\leq t_2\leq \ldots\leq t_k$. Then $\nG(Q)$ must contain
the following blocks: \begin{enumerate}
    \item An $f$-resilient message block between~$\theta_s$ and~$\theta_x$, for each $x \leq k$; and
    \item \label{cond2}
    An $(f-1)$-resilient message block between $\theta_{x-1}^+$ and $\theta_{x}$ that  does not contain null messages sent by $i_{x-1}$, for each $1<x\le k$.
    \end{enumerate}
\end{theorem}

Observe that \Cref{cond2} of \Cref{thm:OR-necc-1} implies the existence of an $(f-1)$-resilient message block in which $i_h$ does not send null messages. This requirement is weaker than the  existence of an $f$-resilient message block. 
While the conditions for \OR\ stated in~\Cref{thm:OR-necc-1}
are necessary, they are not sufficient in general. Indeed, it is unclear what   conditions on $\nG(Q)$ might be sufficient to solve \OR\  for general protocols~$Q$. 
In the Appendix, we present a natural class of protocols for which conditions that are both necessary and sufficient can be stated and proven (see \Cref{thm:OR-necc-2,thm:or-suff}).

\section{Robust Information Transfer}\label{sec:robust-signal}
We now turn to consider protocols that convey information in a ``robust'' way. Namely, they ensure that in every run in which $v_s=1$ and the source process~$s$ does not fail, the destination eventually knows that $v_s=1$. 

\begin{definition}[Robust Information Transfer]
A  protocol~$Q$ is said to solve {\em Robust Information Transfer} between processes~$s$  and~$d$ if, for every run $r$ of $Q$ in which $v_s=1$ and $s$ does not fail, there is a time~$m$ such that  $K_d(v_s=1)$ holds at $(r,m)$.  
\end{definition}

Clearly, a protocol that solves the Robust Information Transfer problem also solves, in particular, \IT\ in its nice run. However, Robust Information Transfer is a strictly harder problem and so, 
as shown in this section, its solutions require more communication than is allowed by $f$-resilient message blocks.

\begin{theorem}[Robust \IT\ Necessity]\label{thm:rob-nec}
Let $Q$ be a protocol that solves Robust \IT\ between $s$ and $d$. Then, 
there exists $m\geq 0$ such that
\begin{itemize}
    \item $\nG(Q)$ contains an actual message sent from~$s$ to $d$ no later than at time~$m-1$, or
    \item  $\nG(Q)$ contains $f+1$ paths from $\theta_s=\node{s,0}$ to $\theta_d=\node{d,m}$ that are disjoint in message senders (except for $s$ and~$d$) such that in at least one of these paths, $s$ does not send null messages. 
\end{itemize}
\end{theorem}

\begin{proofsketch}
    The claim is proved by way of contradiction assuming there exists no $m$ as described in the Theorem. We then consider different cases according to the way the Theorem's assumptions are violated. For each case, we construct a run $r\in R_Q$ in which $v_s=1$ and $s$ does not fail as well as a corresponding run $r'$ in which there is no enhanced message chain from $\theta_s$ to $\theta_d$ and that is indistinguishable by $d$. By \Cref{thm:msg-chain}, it results that $\neg K_d(v_s=1)$ at $(r,m)$, completing the proof.
\end{proofsketch}
An interesting difference between the necessary conditions for solving the Nice-run \IT\ and the ones for Robust \IT\ is in that for the former, process disjointness is required only for the processes sending null messages, while for the latter it is required for all  processes. This resembles the conditions of \cite{dolev1982byzantine} where in order to solve Byzantine Agreement in the synchronous byzantine model, the communication network must have a connectivity of at least $2f+1$. Thus, the paths from the source process~$s$ must be process disjoint the stronger sense. 

We now show that the conditions of \Cref{thm:rob-nec} are not only necessary, but also sufficient.

\begin{theorem}[Robust \IT\ Sufficiency]\label{thm:rob-suff}
Let $\CG$ be a communication graph satisfying the conditions of \Cref{thm:rob-nec}.
Then there is a protocol $Q$ with $\nG(Q)=\CG$ that solves the Robust \IT\ problem.
\end{theorem}

In analogy to how  we defined Nice-based Message protocols in \Cref{sec:niceIT}, we define in the Appendix (\Cref{def:RbM}) what we call Robust-based Message protocols ---  protocols in which whether messages are sent and the contents of the messages sent depend on whether a process knows that $(v_s=0 \,\vee\, \mbox{$s$ failed})$.
Taken together,  \Cref{thm:rob-nec,thm:rob-suff} provide a tight characterization of the communication patterns of protocols solving Robust Information Transfer. 

\section{Conclusions and Future Work}
\label{sec:Conclusions}
Every model of distributed computing provides particular means by which processes can communicate, and these can have a profound impact on the problems that can be solved in the model and on the form that protocols solving them will have. Synchronous systems with global clocks, for example, allow nontrivial use of null messages, which are completely meaningless in the asynchronous model, for example. Since a message not sent can be informative only if there are alternative conditions under which it would be sent, null messages are especially useful as a means of shifting communication costs to optimize for the common case. As illustrated in \Cref{fig:ONet,fig:ONull} an demonstrated in the Atomic Commitment protocols of \cite{silence}, shifting these costs in a careful way can result in significant savings. 

By refining the definitions of null messages, we were able to investigate fundamental aspects of information transfer and coordination in synchronous systems with crash failures. In particular, we obtained characterizations of protocols that solve information transfer and coordination problems in nice, failure-free executions. 
A central tool in our analysis is the  notion of an $f$-resilient message block, which is significantly more refined than the silent choirs of \cite{silence}. Indeed, while constructing  silent choirs is a necessary condition on protocols solving information transfer in nice runs, constructing $f$-resilient blocks is both necessary {\em and} sufficient.  For the Ordered Response coordination problem, where liveness needs to be guaranteed in nice runs, we obtain a condition based on resilient message blocks which, again, is both necessary and sufficient. No similar analysis of Ordered Response has been attempted in the literature. 

There are various synchronous models of current interest in the literature. The  beeping models \cite{CASTEIGTS201920} or  dynamic networks \cite{dynamicNetworks, KFBDynamic}, for example, are intrinsically different from our model. Our results do not directly apply there, but perhaps an analogous analysis will provide insights into such models as well. 

Finally, our work can be extended in several natural ways. One would be to consider  additional failure models. (An early work using the notion of silent choirs to optimize Consensus in the Byzantine setting is~\cite{silence2}.)
Another extension would be to consider more general models of communication. For example, we can replace the assumption that messages take exactly~1 round to be delivered by assuming longer bounds on transmission times. 
 As shown in \cite{ben2014beyond}, assuming an upper  bound on transmission times gives rise to a much richer structure, even in the absence of failures. Doing so in a fault-prone setting such as ours would necessarily lead to generalizations of message chains that extend both resilient message blocks and the centipedes of~\cite{ben2014beyond}. This is a promising direction, which is likely to be both challenging and rewarding.


    

\bibliography{references(2),z1,z2}

\newpage
\section*{Appendix}\label{sec:appendix}
\textbf{Additional content about Failure Patterns defined in \Cref{sec:model}}

Our technical analysis is facilitated by defining a strictness ordering among failure patterns, and associating a minimal failure pattern with each run.

Notice that there may be several failure patterns that are compatible with a given run. In particular, blocking a channel along which no message should be sent does not affect processes' local states.
\begin{definition}[Failure patterns comparison]
Let $FP=\{\langle q_i, t_{i}, Bl(q_i)\rangle\}_{i\le k}$ and $FP'=\{\langle q_j', t_{j}', Bl'(h'_j)\rangle\}_{j\le k'}$ be two failure patterns with $k,k'\leq f$, and denote by $F$ and $F'$ the sets of processes that appear in these patterns, respectively. We say that $FP'$ is {\em harsher}
 than $FP$ (denoted by $FP'\leq FP$) if $F\subseteq F'$ and for every $q_i\in F$:
\begin{itemize}
    \item  $t_i'< t_i$, or
    \item $t_i=t_i'$ and $Bl(q_i)\subseteq Bl'(q_i)$
\end{itemize}
\end{definition}
Note that processes that don't fail in $r'$ might fail in $r$. We now define what we call a {\em minimal} failure pattern wrt.\  a run.
\begin{definition}
Let $r$ be a run and let $FP$ be a failure pattern compatible with~$r$. We say that $FP$ is minimal wrt.~$r$ if for every failure pattern $FP'$ that is compatible with $r$ and such that $FP$ is harsher than $FP'$ (i.e., $FP\leq FP'$), it holds that $FP'=FP$.
\end{definition}
Clearly, for a given run $r$ there is only one minimal compatible failure pattern. We denote it by $FP(r)$. In the proofs appearing in this section,  we will be interested in comparing communication graphs of two runs of the same protocol. We compare the communication graphs regardless of the edges category. Formally:
\begin{definition}[unlabeled-edge subgraph]
Let $\CG=(\nodes,E)$ and $\CG'=(\nodes,E')$ be two communication graphs. We say that $\CG$ is an ``unlabeled-edge'' subgraph of $\CG'$, and write $\CG\subseteq_{u}\CG'$ if for every edge $e\in E$ it is the case that $e\in E'$. (Although $e$ can be an actual message edge in one graph and a null-message edge in the other.) In the rest of paper, we often write ``subgraph'' to stand in for  ``unlabeled-edge subgraph''. 
\end{definition}

We can now show: 

\begin{lemma}\label{lem:subgraph-gen}
        If $FP(r')\leq FP(r)$ for two runs of a protocol~$Q$, then $\CG_Q(r')\subseteq_{u} \CG_Q(r)$.
\end{lemma}
\begin{proof}
Both graphs have the same set of nodes.
We now prove that for each $e\in\Eloc\cup\Eactual(r')\cup\Enull(r')$, it holds that $e\in\Eloc\cup\Eactual(r)\cup\Enull(r)$.
\begin{itemize}
    \item $\Eloc(r)=\Eloc(r')$.
    \item Let $e\triangleq(\node{p,t},\node{p',t+1})\in\Eactual(r')$, i.e., $p$ sends $p'$ an actual message at $(r',t)$. If $p$ sends $p'$ an actual message at $(r,t)$ then $e\in\Eactual(r)$. Otherwise, since $FP(r')\leq FP(r)$, it holds that the channel $\chan{p,p'}$ is not blocked at $(r,t)$. Hence, by the definitions of communication graphs and of null messages, we have that $e\in\Enull(r)$. So $e\in E$.
    \item Finally, let $e\triangleq(\node{p,t},\node{p',t+1})\in\Enull(r')$. In particular, the channel $\chan{p,p'}$ is not blocked at $(r',t)$. Hence by null messages definition and the fact that $FP(r)\leq FP(r')$ it holds that $e\in\Eactual(r)\cup\Enull(r)$.
\end{itemize}
\end{proof}
{\bf Proof of \Cref{thm:f-block-nec}}

In the proof of \Cref{thm:f-block-nec}, we will construct a run $r'$ that is similar to $r$ but in which we make fail additional processes. In order to ensure that these additional failures are not noticed by $d$ (and hence to ensure that $d$ does not distinguish $r$ from $r'$), we make the processes fail at a specific {\em critical} time that we define as follows:

\begin{definition}[Critical Time]
Let~$r$ be a run of $Q$, let $B$ be a set of processes and let~$p$ be a process. 
For every pair $\theta$ and $\theta'$ of nodes of $\CG(Q)$, the {\em critical time} $t_p=t_p(\theta,\theta')$ wrt.\ $(\CG_Q(r),B)$ is defined to be the minimal time $m_p$ such that  $\CG_Q(r)$ contains a $\Bnullfree$ path from~$\theta$ to $\node{p,m_p}$ as well as a path from $\node{p,m_p}$ to $\theta'$. If no such time $m_p$ exists, then $t_p=\infty$.
\end{definition}
Informally, the critical time of a process $p$ represents the first time at which $p$ can learn about an event local to $s$ and may be able to inform $d$ about this event. Making relevant processes fail at their critical times ensures that $d$ does not notice these failures and hence that $d$ does not distinguish the constructed run~$r'$ from the nice run.
We can now prove \Cref{thm:f-block-nec}:

\begin{proof}
Assume by way of contradiction that  no such 
block exists in $\CG_Q(r)$.
I.e., there exists a set~$B$ such that $|B\cup \faulty^r|\leq f$ and every path from  $\theta_s \text{ to } \theta_d$ in $\CG_Q(r)$ contains a null message from a process in $B$. Let $B$ be a minimal set (by set inclusion) with this property.
Define the set $T_B$ to be:
\[T_B\triangleq\{\node{p,t}\in\mathbb{V}:\text{There is no }\Bnullfree \mbox{~path from }\theta_s  \mbox{ to } \node{p,t}\mbox{ in }\CG_Q(r)\}\]

Notice that our assumption about $\theta_d$ implies that  $\theta_d\in T_B$. Moreover, observe that if $\node{i,t}\in T_B$, then $\node{i,t'}\in T_B$ for all earlier times $0\le t'<t$. 

We show that there exists a run $r'$ of~$Q$ such that $r'_{d}(m)=r_{d}(m)$ and there is no enhanced message chain from $\theta_s$ to $\theta_d$ in $r'$. This will  contradict the fact that~$K_{d}(\theta_s\rightsquigarrow\theta_d)$ holds at time~$m$ in~$r$. We construct $r'$ as follows: 
The initial global state is $r'(0)=r(0)$.
Each process~$b\in B$ crashes in~$r'$ at its critical time~$t_b\triangleq t_b(\theta_s,\theta_d)$ wrt. $(\CG_Q(r),B)$ without sending any messages from time $t_b$ on. Moreover, every process in $\faulty^r\backslash B$ crashes in precisely the same manner in $r'$ as it does in $r$.  
By definition, the critical time of a process $p\in\faulty^r$ is necessarily smaller or equal to the actual time at which $p$ fails in $r$.
We hence have that $FP(r')\leq FP(r)$. Clearly, there is no path from $\theta_s$ to $\theta_d$ in $\CG_Q(r')$. Notice that by minimality of $B$, each process $p\in B$ has a finite critical time $t_p=t_p(\theta_s,\theta_d)$ wrt.\ $(\CG_Q(r),B)$.
We now prove by induction on~$t$ that for all $\node{i,t}\in T_B$, if~$i$ has not crashed by time~$t$ in~$r'$, then~$r_{i}(t)=r'_{i}(t)$.

\uline{Base}:~$t=0$. By assumption, $r'(0)=r(0)$.
 Thus, $r'_{i}(0)=r_{i}(0)$ for every process~$i$ and in particular for those satisfying $\node{i,0}\in T_B$. 

\uline{Step}: Let~$t>0$ and assume that the claim holds for all nodes $\node{l,t'}$ with $t'<t$.
Fix a node~$\node{i,t}\in T_B$. Clearly,  $\node{i,t-1} \in T_B$, and so by the inductive hypothesis $r_{i}(t-1)=r'_{i}(t-1)$. 
To establish our claim regarding $\node{i,t}$, it suffices to show that~$i$ receives exactly the same messages at time~$t$ in both runs. 
Since messages are delivered in one time step in our model, the only messages that~$i$ can receive at time~$t$  are  ones sent at time~$t-1$. Hence, we reason by cases, showing that  every process~$z\ne i$ sends~$i$ the same messages at time~$t-1$ in both runs. 
\begin{itemize}
\item Suppose that $\node{z,t-1}\in T_B$. 
\begin{itemize}

    \item If $z\in \faulty^r \backslash B$ then it is active at time~$t-1$ in~$r'$ iff it is active at this time in $r$. We have by the inductive assumption that it has the same local state  in $r$. Since~$Q$ is deterministic, $z$ sends~$i$ a message at time~$t-1$ in~$r'$ iff it does so in~$r$. Moreover, if it sends a message, it sends the same message in both cases. 
    
    \item We show that if $z\in B$, then the channel $\chan{z,i}$ is not blocked at time $t-1$ in $r'$. Assume by way of contradiction that $\chan{z,i}$ is blocked at time $t-1$. This means that $z$  has failed in~$r'$ by time $t-1$. 
    Since $z\in B$ and $z$ has failed in $r'$ by time $t-1$, we have that $t_z\le t-1$. By definition of~$z$'s critical time $t_z$, there is a $\Bnullfree$ path $\pi$ from $\theta_s$ to $\node{z,t_z}$ in $\CG_Q(r)$. There also is a path in $\CG_Q(r)$ from $\node{z,t_z}$ to $\node{z,t-1}$ consisting of locality edges. 
    Together, these two paths form a $\Bnullfree$ path from $\theta_s$ to $\node{z,t-1}$ in $\CG_Q(r)$, contradicting the assumption that $\node{z,t-1}\in T_B$.
    
   \item Assume that $z\in B$ and $\chan{z,i}$ is not blocked at time $t-1$ in $r'$. Then, as in the previous case, the inductive assumption and the fact that~$Q$ is deterministic imply that exactly the same communication occurs between $\node{z,t-1}$ and $\node{i,t}$ in both runs.  

   \item Finally, assume that $z\notin B\cup\faulty^r$. Then $z$ is active at time~$t-1$ in both~$r'$ and $r$. We have by the inductive assumption that it has the same local state  in $r$. Since~$Q$ is deterministic, $z$ sends~$i$ a message at time~$t-1$ in~$r'$ iff it does so in~$r$. Moreover, if it sends a message, it sends the same message in both cases.
\end{itemize}

\item Now suppose that~$\node{z,t-1} \notin T_B$, i.e., there is a $\Bnullfree$ path from $\theta_s$ to $\node{z,t-1}$ in $CG_Q(r)$. Since $\node{i,t}\in T_B$ we have that $z$ does not send a message to $i$ at time $t-1$ in $r$.
There is no edge from $\node{z,t-1}$ to $\node{i,t}$ in $\CG_Q(r)$. Since $FP(r')\leq FP(r)$, it holds by \Cref{lem:subgraph-gen} that $\CG_Q(r')\subseteq_{u} \CG_Q(r)$ and hence, there is no edge from $\node{z,t-1}$ to $\node{i,t}$ in $\CG_Q(r')$ either. Meaning that $i$ does not receive a message from $z$ neither at $(r,t)$ nor at $(r',t)$.
\end{itemize}

Since $r_{i}(t-1)=r'_{i}(t-1)$ process~$i$ performs the same actions at time~$t-1$ in both runs. Since, in addition, $i$ receives exactly the same messages at time~$t$ in~$r'$ as it does in~$\rnice$ as we have shown, it follows that $r_i(t)=r'_i(t)$. 

The inductive argument above showed that, for all processes~$i$ and all times~$t\le m$, if $i$ is active at time~$t$ and  $\node{i,t}\in T_B$, then  $r_i(t)=r'_i(t)$.
Since, $\theta_d\in T_B$ by assumption, it follows that, in particular, 
$r_{d}(m)=r'_{d}(m)$. 
Since there is no enhanced message chain from $\theta_s$ to $\theta_d$ in $r'$, we obtain that $\neg K_{d}(\theta_s\rightsquigarrow\theta_d)$ at time~$m$ in~$r$ by the \Cref{def:know} of the knowledge operator. This contradicts the assumption that $K_{d}(\theta_s\rightsquigarrow\theta_d)$ holds at time~$m$ in~$r$, completing the proof. 
\end{proof}
{\bf Proof of \Cref{lem:strongerthansilentchoir}}
\begin{proof}
Denote by $\faulty^r$ the set of faulty processes in $r$, and assume that the conditions of the Silent Choir Theorem do not hold. I.e.,  neither a silent choir nor an actual message chain from $\theta$ to $\theta'$ exist in $r$. Let $\theta'=\node{q,m}$. Let $S$ be the set of processes such that for each $p\in S$ there is an actual message chain from $\theta$ to $\node{p,m-1}$. Since there is no silent choir, the following holds: $|S\cup \faulty^r|\leq f$.

Let $B\triangleq S\cup \faulty^r$ and let $\pi$ be a path from $\theta$ to $\theta'$ in $CG_Q(r)$. Since there is no actual message chain from $\theta$ to $\theta'$ in $CG_Q(r)$ we get that 
there is an edge from a process of $B$ in $\pi$ that is in $\Enull$, i.e., corresponds to a null message sent by a process in $B$. This holds for every path from $\theta$ to $\theta'$. Hence there exists a set of processes $B$ such that $|B\cup\faulty^r|\leq f$ and such that there is no $\Bnullfree$ path from $\theta$ to $\theta'$ in $CG_Q(r)$, i.e., the conditions of \Cref{thm:f-block-nec} do not hold, completing the proof. 
\end{proof}
{\bf Proof of \Cref{thm:OSN tight-suff}}
\begin{proof}
The assumptions guarantee that there will always be at least one path from~$\theta_s$ to~$\theta_d$ in $\CG$ along which no ``silent'' process fails. Let $Q'$ be an \NbM \ protocol such that $\nG(Q')=\CG$. We show by induction that in all runs  in which~$v_s=0$ each process along the path will detect that the run is not nice. In particular, $j$ will be able to distinguish the run from the nice one by time~$m$. It follows that $K_{d}(v_s=1)$ holds in~$\rnice$ at time~$m$. 

Let $r'$ be a run in which $v_s=0$ and denote by $B$ the set of processes that fail in this run. Clearly, $|B|\leq f$. Let~$\pi$ be a $\Bnullfree$ path in $\nG(Q')$, which is guaranteed to exist by the assumption. Let $r'$ be a run in which $v_s=0$. We now prove by induction on time that for each node $\node{p,t}$ in~$\pi$ it holds that~$K_p(\neg \nicerun)$ holds at $(r',t)$. . 

\uline{Base:}~$t=0$. In this case, $p=s$. Since $v_s$ appears in $s$'s local state and its value differs to its value in the nice run, $K_{s}(\neg \nicerun)$ holds at time 0.

\uline{Step:}~$t>0$. We consider the nodes~$\node{q,t-1}$ and~$\node{p,t}$ in~$\pi$. By the induction hypothesis~$K_q(\neg \nicerun)$ holds at $t-1$ in $r'$. We now reason by cases according to the class of the edge~$(\node{q,t-1},\node{p,t})$ in~$\nG(Q')$.
\begin{itemize}
    \item Case 1: $\left(\node{q,t-1},\node{p,t}\right)\in \Eloc$ then $p=q$ and since the fact $\neg\nicerun$ is a stable property we have by the induction hypothesis that $K_p(\neg \nicerun)$ holds at $(r',t)$.
    \item Case 2:~$\left(\node{q,t-1},\node{p,t}\right)\in \Eactual(\rnice)$:
    \begin{itemize}
        \item Case 2a: in the run~$r'$ process~$q$ does not send $p$ a message, then~$p$ detects that the run is not~$\rnice$ (in which, by assumption, it would receive a message from~$q$). 
        \item Case 2b: $q$ does send a message to~$p$ in~$r'$ then, by the induction assumption and the fact that $q$ sends~$0$ if $K_q(\neg\nicerun)$ it follows that~$p$ receives a different message in $r'$ and in~$\rnice$, and so $K_p(\neg\nicerun)$ holds at time~$t$.
    \end{itemize}
    \item Case 3: ~$\left(\node{q,t-1},\node{p,t}\right)\in \Enull(\rnice)$ we have by the choice of~$\pi$ that~$q$ does not fail in $r'$ and by the induction assumption~$K_q(\neg \nicerun)$ holds at time~$t-1$. Recall that, by assumption, in~$Q'$ process~$q$ can send a null message only in case $\neg(K_q\neg\nicerun)$. Since, by the inductive assumption on time~$t-1$ this is not the case, $q$ must send~$p$ a `0'-message. Since such messages are never sent in $\rnice$, we again conclude that~$K_p(\neg \nicerun)$ holds at time~$t$ in~$r'$, as desired.
\end{itemize}
We have shown that for all runs~$r'$ in which $v_s\ne 1$ it is the case that $r'_d(m)\ne \rnice_d(m)$. 
Consequently, $v_s=1$ for all runs~$r$ such that $r_d(m)=\rnice_d(m)$ and so, by  \Cref{def:know}, we obtain that  $K_{d}(v_s=1)$ holds at $(\rnice,m)$, as claimed. 
\end{proof}
{\bf Proof of \Cref{thm:rob-nec}}
\begin{proof} Assume by way of contradiction that the claim is false, and let $Q$ be a protocol for which $\nG(Q)$ does not satisfy the conditions. Recall that in $Q$'s nice run $\rnice$ both $v_s=1$ and process~$s$ does not fail.  
Let $m\geq 0$. First, since $s$ does not send a real message to $d$ by time $m-1$ in $\rnice$, it holds that, in $\rnice$, every enhanced message chain only containing messages sent by $s$ contains a null message sent by $s$.
In addition, by assumption there are at most $f-1$ other paths disjoint in message senders that do not contain null messages from $s$. 
As in the proof of \Cref{thm:f-block-nec}, we can construct a run $r$ in which $s$ and at most $f-1$ other processes fail in a way that there is no enhanced message chain from $\theta_s=\node{s,0}$ to $\theta_d=\node{d,m}$ in $r$ and $r_d(m)=\rnice_d(m)$. By \Cref{thm:msg-chain},  $K_d(v_s=1)$ does not hold at $(\rnice,m)$ and hence, $Q$ does not solve Robust \IT. Otherwise, we have that every enhanced message chain (in particular every actual message chain) from $\theta_s$ to $\theta_d$ contains  messages sent by some process different than $s$, and in addition:  either (i) there is no set of~$f+1$ paths from $\theta_s$ to $\theta_d$ in $\nG(Q)$ that are disjoint in messages senders or (ii) for every set of paths in $\nG(Q)$ of size $f+1$ between $\theta_s$ and $\theta_d$ that are disjoint in messages senders, process $s$ sends a null message in (at least) two of these paths. Let $\Pi$ be the set of paths from $\theta_s$ to $\theta_d$ in $\nG(Q)$.
\begin{itemize}
    \item Clearly, in case (i) there exists a set $B$ of~$f$ processes that such that $s\notin B$ and such that for every $\pi\in\Pi$ there is a process $p\in B$ that sends a message (may be actual or null) along $\pi$. We construct a run $r$ in which $v_s=1$ and each process of $B$ fails at the first time it is reached by an enhanced message chain from $\theta_s$ in $\nG(Q)$. When failing, all the channels originating from the process are blocked at the time it fails. In particular, $r$ is a run in which $v_s=1$ and $s$ does not fail. By the assumption on $B$ and by applying \Cref{lem:subgraph-gen} on $\CG_Q(r)$ and $\nG(Q)$, it results that there is no enhanced message chain from $\theta_s$ to $\theta_d$ in $r$. Hence, by \Cref{thm:msg-chain}, $K_d(v_s=1)$ does not hold at $(r,m)$.

    \item Now for case (ii) let $r$ be a run in which $v_s=1$, process~$s$ does not fail, and $f-1$ other processes fail in such a way  that every remaining path from $\theta_s$ to $\theta_d$ in $\CG_Q(r)$ contains a null message sent by $s$. This run is a legal run of $Q$ since for every set of $f+1$ paths between $\theta_s$ and $\theta_d$ there are at least two paths that contain null messages from $s$ and applying \Cref{lem:subgraph-gen} on $\CG_Q(r)$ and $\nG(Q)$ ensures that $\CG_Q(r)$ contains no additional path. We define the run $r'$ to be the run in which $s$ fails at time $0$ , all the channels originating from $s$ are blocked at time $0$ and the remaining processes fail according to the same failure pattern as in~$r$.
Hence, $FP(r')$ is harsher than $FP(r)$, i.e., $FP(r')\leq FP(r)$ and by \Cref{lem:subgraph-gen} it is the case that $\CG_Q(r')\subseteq_{u} \CG_Q(r)$.
We now show by induction on time that for each node $\node{p,t}$ from which there is a path to $\theta_d$ in $\CG_Q(r')$ it is the case that $r_p(t)=r'_p(t)$.
Define the set $T$ to be: \[T\triangleq\{\node{p,t}:\node{p,t}\rightsquigarrow\theta_d\text{ in } r'\}\]

\uline{Base:} $t=0$: Each process has the same local state in both $r$ and $r'$ so this holds in particular for nodes $\node{p,0}\in T$.

\uline{Step:} $t>0$. Let $\node{p,t}\in T$ and let $\node{z,t-1}$ nodes of $\CG_Q(r')$. We reason by cases.
\begin{itemize}
    \item Case 1: $\node{z,t-1}\in T$ then  $z\neq s$. (Because $z$ being equal to $s$ would lead to a contradiction to the fact that there is no path from $\theta_s$ to $\theta_d$ in $\CG_Q(r')$). By the inductive hypothesis we have that $z$ has the same local state at $(r,t-1)$ and $(r',t-1)$. Now, depending on the class of edge $(\node{z,t-1},\node{p,t})$ we have:
    \begin{itemize}
        \item Case 1a: $(\node{z,t-1},\node{p,t})\in\Eactual$. Since a process that fails in $r$ fails in the same manner in $r'$, it holds that $z$ sends $p$ the same message at both $(r,t-1)$ and $(r',t-1)$.
        \item  Case 1b: $(\node{z,t-1},\node{p,t})\in\Enull$. Recall that $r_{z}(t-1)=r'_{z}(t-1)$. Hence, $p$  receives a message from $z$ neither at $(r,t)$ nor at $(r',t)$.
    \end{itemize}
    \item Case 2: $\node{z,t-1}\notin T$. This means that there is no edge from $\node{z,t-1}$ to $\node{p,t}$ in $\CG_Q(r')$. In particular, no message is sent by $z$ to $p$ at $(r',t-1)$. We show that in $r$, no such message is sent either. Now, depending on if $z$ is equal to $s$ we have:
    \begin{itemize}
        \item Case 2a: $z\neq s$: Since there is no edge from $\node{z,t-1}$ to $\node{p,t}$ in $\CG_Q(r')$ it results that no edge from $\node{z,t-1}$ to $\node{p,t}$ exists either in $\CG_Q(r)$. This results from the definition of a communication graph and the fact that processes other than $s$ fail in the same way in $r$ and $r'$. I.e., no message is sent from $z$ to $p$ neither at $(r',t)$ nor at $(r,t)$.
        \item Case 2b: $z=s$. Clearly, since $s$ failed at $(r',0)$ no message is sent by $s$ to $p$ at $(r',t-1)$. Recall that every path in $\CG_Q(r)$ from $\theta_s$ to $\theta_d$ contains a null message sent by $s$. Assume by way of contradiction that $s$ sends $p$ an actual message at $(r,t-1)$. We hence get a path $\pi_1$ in $\CG_Q(r)$ consisting of local edges from $\theta_s$ to $\node{s,t-1}$ and of an actual message edge from $\node{s,t-1}$ to $\node{p,t}$. Let $\pi_2'$ be a path from $\node{p,t}$ to $\theta_d$ in $\CG_Q(r')$ as guaranteed to exist by the choice of $\node{p,t}$. Clearly, $\pi_2'$ does not contain messages sent by $s$ (since $s$ has failed at time $0$). By \Cref{lem:subgraph-gen} there is a path $\pi_2$ in $\CG_Q(r)$ from $\node{p,t}$ to $\theta_d$ with no messages sent by $s$. Concatenating $\pi_1$ with $\pi_2$, we get a path from $\theta_s$ to $\theta_d$ in $\CG_Q(r)$ that does not contain null messages sent by $s$. Contradicting the assumption.
    \end{itemize}
\end{itemize}
\end{itemize}

To conclude, we showed that there exists a run $r'$ in which both  $r_d(m)=r'_d(m)$ and  $\CG_Q(r')$  contains   no path from $\theta_s$ to $\theta_d$. Hence, $\neg K_d(\theta_s\rightsquigarrow\theta_d)$ at $(r,m)$. So by \Cref{thm:msg-chain}, we have that $\neg K_d(v_s=1)$ at $(r,m)$, despite the fact that $v_s=1$ and $s$ does not fail in $r$. Contradicting the assumption about~$Q$.
\end{proof}
\textbf{Definition of Robust-based Message protocols}
\begin{definition}[Robust-based Message protocols]\label{def:RbM}
 We say that~$Q$ is a Robust-based Message (\RbM) protocol if \begin{itemize}
    \item All actual messages sent in~$Q$ are single-bit messages, and whenever a process~$p$ sends an actual message, it sends a~`0' if $K_{p}(v_s\neq 1 \vee s \text{ is faulty})$ and sends a~`1' otherwise,
    \item for all processes~$p$, each null message sent by $p$ over any channel is a null message in case $\varphi=\neg K_p[(v_s\neq 1)\vee (s \text{ is faulty})]$ and
    \item for every run $r$ and process $p$ it holds that if $p$ sends $q$ an actual message at $(\rnice,t)$, then if the channel $\chan{p,q}$ is not blocked at $(r,t)$, process $p$ also sends $q$ an actual message at $(r,t)$.
\end{itemize}  
\end{definition}

{\bf Proof of \Cref{thm:rob-suff}}
\begin{proof}
Let $\CG$ be a communication graph as described, and let ~$Q$ be an \RbM\  protocol such that $\nG(Q)=\CG$.

Let $r$ be a run of $Q$ in which $v_s=1$ and $s$ does not fail. Let $r'$ be a run in which $v_s=0$ and denote by $\faulty^{r'}$ the set of faulty processes in $r'$. We will show that $r_d(m)\neq r'_d(m)$. 

First, assume there is an actual message chain from $\theta_s$ to $\theta_d$ along which $s$ is the only process sending messages. Then, no matter whether~$s$ fails or not in $r'$, we clearly have that $r_d(m)\neq r'_d(m)$.

Otherwise, the second condition holds. We now reason by cases:

\begin{enumerate}
    \item Case 1: $s\in\faulty^{r'}$: The conditions on the graph imply that there is (at least) one path $\pi$ in $\nG(Q)$ from $\theta_s$ to $\theta_d$ that does not contain messages (neither null nor actual) sent by processes in $\faulty^{r'}\backslash\{s\}$, and in addition, $\pi$ does not contain null messages from $s$. We now show by induction on time that for each node $\node{p,t}\in\pi$, it is the case that $K_p(v_s\neq 1\vee s \text{ is faulty})$ holds at $(r',t)$.
    
    \uline{Base:} $t=0$. In this case, $p=s$. Since $v_s$ appears in $s$'s local state, $K_{s}(v_s\neq 1)$ holds at time 0.

    \uline{Step:}~$t>0$. We consider the nodes~$\node{q,t-1}$ and~$\node{p,t}$ in~$\pi$. By the induction hypothesis~$K_q(v_s\neq 1\vee s \text{ is faulty})$ holds at $t-1$ in $r'$. We now reason by cases according to the class that the edge~$(\node{q,t-1},\node{p,t})$ belongs to  in~$\CG_Q(r)$.
    \begin{itemize}
    \item Case 1a:~$\left(\node{q,t-1},\node{p,t}\right)\in \Eloc$ then $p=q$ and since the fact $(v_s\neq 1\vee s \text{ is faulty})$ is a stable property we have by the induction hypothesis that $K_p(v_s\neq 1\vee s \text{ is faulty})$ holds at $(r',t)$.
    \item Case 1b:~$\left(\node{q,t-1},\node{p,t}\right)\in \Eactual(r)$:
    \begin{itemize}
        \item Case 1b (i): $q$ is active at $(r',t-1)$, it holds by $Q$'s definition and by the induction assumption that $K_p(v_s\neq 1\vee s \text{ is faulty})$ holds at $(r',t)$.
        \item Case 1b (ii): $q$ is not active, then $q=s$ (by the choice of $\pi$). Hence, $K_p(v_s\neq 1\vee s \text{ is faulty})$ holds at $(r',t)$.
\end{itemize}
\item Case 1c: $\left(\node{q,t-1},\node{p,t}\right)\in \Enull(r)$ we have by the choice of~$\pi$ that~$q$ does not fail in $r'$ and by the induction assumption~$K_q(v_s\neq 1\vee s \text{ is faulty})$ holds at $(r',t-1)$. Recall that, by assumption, in~$Q$ process~$q$ can send a null message only in case \\ $\neg K_q[(v_s\neq 1)\vee (s \text{ is faulty})]$. Since, by the inductive assumption on time~$t-1$ this is not the case, $q$ must send~$p$ a message and since the protocol is \RbM, $K_p(v_s\neq 1\vee s \text{ is faulty})$ holds at time~$t$ in~$r'$, as desired.

\end{itemize}
To conclude, we have shown that for every run~$r'$ in which $v_s\ne 1$ it is the case that $r'_d(m)\ne r_d(m)$.
Consequently, $K_{d}(v_s=1)$ holds at $(r,m)$, as claimed. 
    \item Case 2: $s\notin\faulty^{r'}$ then there is a path in $\CG_Q(r)$ that does not contain messages sent by processes in $\faulty^{r'}$.
    As in the previous case, we can show that $r_d(m)\neq r'_d(m)$.
\end{enumerate}
\end{proof}
\textbf{Ordered Response}
\\\\
{\bf Proof of \Cref{thm:OR-necc-1}}
\begin{proof}
\begin{enumerate}

    \item By \Cref{lem:OR-basic}, we have that $(R_Q,r,t_x)\sat K_{i_h}(\factOR)$ holds for each~$x \leq k$. The necessity part of \Cref{thm:OSN tight} implies that there is an~$f$-resilient message block between~$\theta_s$ and~$\theta_x$ for each~$x \leq k$, as claimed.
    
    \item Recall that since $Q$ solves \OR, for all $x>1$ we have by \Cref{lem:OR-basic} that  $K_{i_{x}}(\ba_{x-1})$ must hold at time~$t_{x}$ in $Q$'s nice run~$\rnice$. Assume by way of contradiction that every set of paths~$\Gamma$ from $\theta_{x-1}^+$ to $\theta_x$ that does not contain null messages from $i_{x-1}$ is not an $(f-1)$-resilient message block. This means that for every set of paths $\Gamma$ from $\theta_{x-1}^+$ to $\theta_x$ that does not contain null messages from $i_{x-1}$, there is a set $B'$ of size $|B'|\leq f-1$ such that there is no $\Btagnullfree$ path in $\Gamma$. Let $\Gamma$ be the set of paths from $\theta_{x-1}^+$ to $\theta_x$ such that no $\pi\in\Gamma$ contains null messages from $i_{x-1}$. Let $B'$ be a set of processes as guaranteed to exist according to the contradiction assumption. We denote $B=B'\cup\{i_{x-1}\}$. Let $T_B$ be the set \[T_B~\triangleq~ \{\node{p,t}\in\mathbb{V}:\text{There is no }\Bnullfree \mbox{~path from }\theta_{x-1}^+  \mbox{ to } \node{p,t}\mbox{ in }\nG(Q)\}\] 
     We first show that $\theta_{x}\in T_B$.
    Let $\pi$ be a path from $\theta_{x-1}^+$ to $\theta_x$. If $\pi$ contains a null message from $i_x$ then $\pi$ is not a $\Bnullfree$ path. Otherwise, by the contradiction assumption, there is a process $b'\in B'\subseteq B$ such that $\pi$ contains a null message from $b'$, i.e., $\pi$ is not a $\Bnullfree$ path. Since this holds for every path between $\theta_{x-1}^+$ and $\theta_x$ we get that $\theta_x\in T_B$. Observe that $|B|\leq f$. We claim that there exists a run $r'$ of~$Q$ in which~$a_{x-1}$ is not performed such that~$r'_{i_{x}}(m)=\rnice_{i_{x}}(m)$. This will contradict the fact that~$K_{i_{x}}(\ba_{x-1})$ holds at time~$t_{x}$ in~$\rnice$. We construct $r'$ as follows: 
    The  global states of~$r'$ and~$\rnice$ satisfy that $r'(m)=\rnice(m)$ for all $m\le t_{x-1}$. \big(Notice that an action performed at time~$t_{x-1}$ may affect the global state only from time $t_{x-1}+1$ on.\big) 
    In addition, $i_{x-1}$ fails at time $t_{x-1}$ without executing $a_{x-1}$ but it does send the messages it sends in $\rnice$ and each process~$b\in B$ such that $b\neq i_{x-1}$ crashes in~$r'$ at its critical time~$t_b(\theta_{x-1}^+,\theta_{x})$ wrt.\ $(\nG(Q),B)$ without sending any messages from time $t_b$ on nor executing actions. The rest of the proof proceeds exactly as in \Cref{thm:OSN tight}, showing by induction on time~$t$ that for all nodes  $\node{i,t}\in T_B$, if~$i$ has not crashed by time~$t$ in~$r'$, then~$\rnice_{i}(t)=r'_{i}(t)$. Since $\theta_{x}\in T_B$, we thus obtain that $\rnice_{i_{x}}(t_{x})=r'_{i_{x}}(t_{x})$. Since $a_{x-1}$ is not performed in~$r'$, the claim follows.

\end{enumerate}
\end{proof}

\begin{definition}[Conservative \OR \ protocols]
Let~$Q$ be a deterministic protocol that solves~$\ORi=\angles{\factOR,a_1,\ldots,a_k}$. We say that~$Q$ is conservative for \ORi\ if for every run~$r$ of~$Q$ and all~$x \leq k$ the following is true:
Process~$i_x$ performs~$a_x$ at~$t_x$ only if~$\neg K_{i_x}(\neg\nicerun)$ holds at~$(r,t_x)$.
\end{definition}

In a conservative protocol, if a process~$i_x$ knows at $\theta_x=\node{i_x,t_x}$ that a failure has occurred, then it is not allowed to perform its action. 
Concretely, suppose that process~$i_x$ is prevented from action at $\theta_x=\node{i_x,t_x}$ because it observes there that a process~$b$ has failed at $\rho_b=\node{b,m_b}$.  
Since by \Cref{thm:OR-necc-1}, only $(f-1)$-resilient message blocks are required between two consecutive processes in the \ORi\ instance,  the failure of $f-1$ other processes might disconnect $\theta_{x}$ from a node~$\theta_{h}$, for an index $h>x$ in the instance of \OR\ being solved. Moreover, this might also disconnect $\rho_b$ from $\theta_{h}$. We would then obtain that $i_{h}$ does not distinguish the current run from the nice run, resulting in~$a_{h}$ being performed. This is clearly a violation of \OR. We illustrate this scenario in \Cref{fig:nec-conservative}.
\begin{figure}
    \centering
    \begin{tikzpicture}
  \node[circle, draw] (rho_b) at (0, 0) {$\rho_b$};
  \node[above] at (rho_b.north) {$=$};
  \node[above, yshift=10pt] at (rho_b.north) {$\langle b, m_b \rangle$};
  \node[circle, draw] (theta_h) at (14, 0) {$\theta_h$};
  
  \node[circle, draw] (rho_q) at (1.5, 0) {$\rho_q$};
  \draw[->] (rho_b) -- (rho_q);
  \draw[->, bend left=30, decorate, decoration={snake, amplitude=1.0mm, segment length=3mm}] (rho_q) --(5.6,0);
  
  \node[circle, draw] (theta_x+1) at (8, 0) {$\theta_{x+1}$};

  \node[circle, draw] (theta_x) at (6, 0) {$\theta_{x}$};
  
  \node[right=0.2cm of theta_x] (dots) {$\cdots$};
  \node[right=0.5cm of theta_x+1] (dots) {$\cdots$};

  \node[draw, ellipse, minimum height=3cm, minimum width=1.5cm, rounded corners, right=0.5cm and 0.2cm of dots] (ellipse) {};
  \node[yshift=12pt, xshift=8pt] at (ellipse) {$\langle b, l \rangle$};
  \node[above left=0.1cm and 0.05cm of ellipse] {$B$};
  \node[below =0.1cm of ellipse] {$(f-1)$-resilient};
  \node[below =0.5cm of ellipse] {msg block};
  \node[below=0.1cm of theta_x] {$a_x$ does not occur};
   \tikzset{dashedarrow/.style={->, dashed, decorate}}
      \draw[ dashedarrow] (ellipse) -- (12.85,0);
  \draw[dashedarrow] (ellipse) -- (12.85, -1.5);
  \draw[dashedarrow] (ellipse) -- (12.85, 1.5);

  \draw[->, bend left=30, decorate, decoration={snake, amplitude=1.2mm, segment length=3mm}] (13.2,0) -- (13.6,0);
  \draw[->, bend left=30, decorate, decoration={snake, amplitude=1.2mm, segment length=3mm}] (13.2,-1.5) -- (14,-0.4);
  \draw[->, bend left=30, decorate, decoration={snake, amplitude=1.2mm, segment length=3mm}] (13.2,1.5) -- (14,0.4);
  
  \draw[red, line width=1.5pt] (12,-1.1) ++(-0.12,0.12) -- ++(0.25,-0.25);
  \draw[red, line width=1.5pt] (12,-1.1) ++(-0.12,-0.12) -- ++(0.25,0.25);

  \draw[red, line width=1.5pt] (12,-0.3) ++(-0.12,0.12) -- ++(0.25,-0.25);
  \draw[red, line width=1.5pt] (12,-0.3) ++(-0.12,-0.12) -- ++(0.25,0.25);

  \draw[red, line width=1.5pt] (0.5,-0.2) ++(-0.12,0.12) -- ++(0.25,-0.25);
  \draw[red, line width=1.5pt] (0.5,-0.2) ++(-0.12,-0.12) -- ++(0.25,0.25);

\node[circle, draw] at (13, 0) {};
\node[circle, draw] at (13, -1.5) {};

\node[circle, draw]  at (13, 1.5) {};
\end{tikzpicture}
    \caption{The problematic scenario that \Cref{thm:OR-necc-2} solves. Squiggly arrows represent any kind of message chain or sets of message chains. Red crosses represent process failures. As in previous figures, dashed arrows represent null messages and full arrows represent real messages.
    If the depicted scenario occurs, then $b$'s failure at $\rho_b$ can cause $i_x$ not to perform $a_x$. The protocol must therefore provide an $(f-1)$-resilient message block to $\theta_h$ from one of $\theta_x$ or $\node{b,m_b+1}$.}
    \label{fig:nec-conservative}
\end{figure}
We can show that in order to prevent such a scenario there must be a $\Bnullfree$ path from $\theta_x$ or from $b$ after its potential failing node~$\rho_b$, to~$\theta_h$. This way, if~$b$ fails at $\rho_b$ and prevents~$i_x$ from acting, then~$i_h$ will distinguish the current run from the nice run at~$\theta_h$. Acting conservatively, $i_h$ will also refrain from acting, and thus avoid causing a violation of the \OR\ specification. Formally:

\begin{theorem}\label{thm:OR-necc-2}

Let~$Q$ be a {\em conservative} protocol
solving $\ORi=\angles{\factOR,a_1,\ldots,a_k}$. 

For all nodes $\rho_b=\node{b,m_b}$, indices $x<h\le k$ and sets $B\subseteq\Proc$, 
if 
\begin{enumerate}
    \item there is a path $\pi$ from $\rho_b$ to~$\theta_x$ in $\nG(Q)$ that starts with an edge $(\rho_b,\rho_q)\in \Eactual$ and contains no edges corresponding to null message by~$b$, and in addition
    \item  $b\in B$, $|B|\le f$ and there is no $\Bnullfree$ path from $\theta_x$ to $\theta_h$ in $\CG$,
\end{enumerate}

then there is a $\Bnullfree$ path from~$\node{b,m_b+1}$ to~$\theta_h$ in $\nG(Q)$.

\end{theorem}
In a precise sense, combining the  conditions in this theorem with those of \Cref{thm:OR-necc-1} we obtain a set of conditions that is not only necessary for conservative protocols~$Q$ (as already proved), but also sufficient. Indeed, as we now show, there exist protocols solving Ordered Response that satisfy precisely these conditions.
\begin{proof}
    Assume, by way of contradiction, that assumptions (1) and (2) hold but there is no $\Bnullfree$ path from~$\node{b,m_b+1}$ to~$\theta_h$ in $\nG(Q)$.
        Let $x\le h<k$, let $B$ a set of processes such that $|B|\leq f$ and let~$b\in B$ such that there is a path from $\rho_b=\node{b,m_b}$ to $\theta_x$ that starts with an edge $\node{\rho_b,\rho_q)}\in \Eactual$ that does not contain null messages sent by $b$ in $\nG(Q)$. In addition there is no $\Bnullfree$ path neither from~$\theta_x$ to~$\theta_h$ nor from~$\node{b,m_b+1}$ to~$\theta_h$.
    The idea is to construct a run $r'$ in which $i_{x}$ knows the current run is not nice and since $Q$ is a conservative protocol, it follows that $a_{x}$ does not occur. In addition, the processes of $B$ fail in a way that $i_{h+1}$ does not differentiate $r'$ from $\rnice$; contradicting the fact that $K_{i_h}(\ba_{x})$ holds at $(\rnice,t_{h})$.
    We construct the run $r'$ to be  identical to $\rnice$ up to and including time~$m_b$ at which point~$b$ fails without sending any messages along the paths reaching $\theta_{x}$ ($b$ does send the other messages it is supposed to send at time~$m_b$).
    In addition, in $r'$ each process $b'\in B$ such that $b'\neq b$ fails at its critical time $t_{b'}=t_{b'}(\theta_x,\theta_h)$ wrt.\ $(\nG(Q),B)$ without sending any messages. Let
    \begin{align*}
        T_B=\{\node{p,t}\in\mathbb{V}:\text{There is neither a }\Bnullfree \mbox{~path from }\theta_x  \mbox{ to } \node{p,t}\mbox{ in }\nG(Q) \\ \mbox{ nor a } \Bnullfree \mbox{~path from }\node{b,m_b+1} \mbox{ to } \node{p,t}\} 
    \end{align*}
    An argument analogous to that in \Cref{thm:OSN tight} now shows that $K_{i_{x}}(\neg \nicerun)$ holds at time~$t_{x}$ and then $i_{x}$ does not act.
    Observe that by the assumption, $\theta_h\in T_B$. Moreover, the same argument as in the proof of \Cref{thm:OSN tight} now shows by induction on~$t$ that for each node $\node{i,t}\in T_B$, it holds that $r'_{i}(t)=\rnice_{i}(t)$. In particular, since  $\theta_h\in T_B$ we can conclude that $\rnice_{i_{h}}(t_{h})=r'_{i_{h}}(t_{h})$, contradicting the fact that $K_{i_{h}}(\ba_{x})$ holds at $(\rnice,t_{h})$ as required by \Cref{lem:OR-basic}. 
\end{proof}
\begin{theorem}[Sufficient conditions for \OR]\label{thm:or-suff}
    The conditions stated in \Cref{thm:OR-necc-1} and \Cref{thm:OR-necc-2} are sufficient for solving an instance $\ORi=\angles{\factOR,a_1,a_2,\ldots,a_k}$ of the Ordered Response problem. 
   Namely, suppose that 
     for some sequence of times $t_1\le t_2\le\cdots\le t_k$  a communication graph $\CG$ contains, for every $h\leq k$, 
    \begin{enumerate}
        \item\label{it:OR-suff-1} $f$-resilient messages blocks between $\theta_s$ and $\theta_h$, and
        \item\label{it:OR-suff-2} an $(f-1)$-resilient message block between $\theta_{h}^+$ and $\theta_{h+1}$ that does not contain null messages from $i_{h}$ for $h<k$. Moreover, assume that 
        \item \label{it:OR-suff-3} for all nodes $\rho_b=\node{b,m_b}$, indices $x<h\le k$ and sets $B\subseteq\Proc$, if 
        \begin{enumerate}
            \item there is a path $\pi$ from $\rho_b$ to~$\theta_x$ in $\CG$ that starts with an edge $(\rho_b,\rho_q)\in \Eactual$ and contains no edges corresponding to null messages by~$b$, and 
            \item it holds in addition that $b\in B$, $|B|\le f$ and there is no $\Bnullfree$ path from $\theta_x$ to $\theta_h$ in $\CG$,  
        \end{enumerate} 
        then there is a $\Bnullfree$ path from~$\node{b,m_b+1}$ to~$\theta_h$ in $\CG$.
    \end{enumerate}
    Then there exists a protocol $Q$ that solves $\ORi$ with respect to  times $t_1\le t_2\le\cdots\le t_k$ such that $\nG(Q)=\CG$.
\end{theorem}
    
    Taken together, \Cref{thm:OR-necc-1,thm:OR-necc-2,thm:or-suff} provide a  characterization of the communication patterns that can solve Ordered Response using null messages. This characterization is tight for communication patterns of {\em conservative} protocols that solve \OR.
\\\\
{\bf Proof of \Cref{thm:or-suff}}

\begin{proof} 
Let $\ORi=\angles{\factOR,a_1,\ldots,a_k}$ be an instance of an \OR \ problem. We define~$Q$ to be an \NbM\ protocol such that $\nG(Q)=\CG$ for some sequence of times $t_1\le t_2\le\cdots\le t_k$. In addition in every run $r$ such that $\neg K_{i_h}( \neg\nicerun)$ holds at $(r,t_h)$, $Q$ instructs process $i_h$ to perform $a_h$ at $t_h$.

Let~$h\leq k$.
We prove that 
\begin{itemize}
    \item $K_{i_h}(\factOR)$ holds at $(\rnice,t_h)$,
    \item $K_{i_{h+1}}\;\!\ba_{h}$ holds at $(\rnice,t_{h+1})$ for every $h<k$  
\end{itemize}

Let $r'$ be a run in which $v_s\neq 1$ and let $h\leq k$. \Cref{it:OR-suff-1} of \Cref{thm:or-suff} states  that there is an~$f$-resilient message block from~$\theta_s$ to~$\theta_h$. Hence, by \Cref{thm:OSN tight} we have that~$r'_{i_{h}}(t_{h})\neq \rnice_{i_{h}}(t_{h})$. Thus,~$K_{i_h}(v_s=1)$ holds at~$(\rnice,t_h)$.
Now, let $r'$ be a run in which $v_s=1$ and in which $i_{h}$ does not execute~$a_h$ for some $h<k$. We prove that $r'_{i_{h+1}}(t_{h+1})\neq \rnice_{i_{h+1}}(t_{h+1})$.
Since~$v_s=1$ there are two possibilities:
\begin{enumerate}
    \item $i_h$ fails at a time $t\leq t_h$ without performing $a_h$. \Cref{it:OR-suff-2} states that there is an $(f-1)$-resilient message block from $\theta_{h}^+$ to $\theta_{h+1}$ that does not contain null messages from $i_h$. It implies that there is a path in $\nG(Q)$ in which no silent process fails in $r'$. Thus, we can show as in \Cref{thm:OSN tight} that $\rnice_{i_{h+1}}(t_{h+1})\neq r'_{i_{h+1}}(t_{h+1})$.
    \item $i_{h}$ does not fail up to and including time~$t_h$. Since $a_h$ is not performed, it follows by $Q$'s definition that $K_{i_{h}}(\neg \nicerun)$ holds at time $t_h$ in $r'$.
    We separate into 2 cases:
    \begin{enumerate}
        \item There is a path~$\pi$ in~$\nG(Q)$ from~$\theta_h^+$ to~$\theta_{h+1}$ such that in~$r'$ no process that sends a null message in~$\pi$ fails. Then, as in \Cref{thm:OSN tight} we can show by induction on time that~$r'_{i_{h+1}}(t_{h+1})\neq\rnice_{i_{h+1}}(t_{h+1})$.
        \item In every path of $\nG(Q)$ from~$\theta_h^+$ to~$\theta_{h+1}$ there is a process that sends a null message that fails. Then it means that by~$t_{h+1}$ there have been~$f$ failures (this results from the fact that for every set $B$ of size $|B|\leq f-1$ there is a $\Bnullfree$ path between $\theta_h^+$ and~$\theta_{h+1}$). Denote by~$B=\{b_1,b_2,\ldots,b_f\}$ the processes that have failed by time~$t_{h+1}$.
        Recall that the fact that~$i_h$ does not act at time~$t_h$ implies that $K_{i_{h}}(\neg \nicerun)$ holds at time $t_h$ in $r'$.
        We look at the smallest~$1\leq j \leq h$ for which $a_j$ has not been performed at~$t_j$ in~$r'$ (In particular it may be that~$j=h$). Recall that $v_s=1$ in $r'$. By the choice of~$j$ and the protocol, it follows that there is a path in $\nG(Q)$ starting by an edge $(\rho_b,\rho_q)\in\Eactual$ (denote $\rho_b=\node{b,m_b}$)to~$\theta_j$ in $\nG(Q)$ such that $b$ fails by time $m_b$ (Otherwise the failure would not be detected). By condition 3, we have that there is at least a $\Bnullfree$ path~$\pi$ in~$\nG(Q)$ from~$\theta_j^+$ to~$\theta_{h+1}$ or a $\Bnullfree$ path~$\pi$ in~$\nG(Q)$ from~$\node{b,m_b+1}$ to~$\theta_{h+1}$. Then, as in \Cref{thm:OSN tight} we can show by induction on time that~$r'_{i_{h+1}}(t_{h+1})\neq\rnice_{i_{h+1}}(t_{h+1})$.
    \end{enumerate}
\end{enumerate}
We hence have shown that both 
\begin{itemize}
    \item $K_{i_h}(\factOR)$ holds at $(\rnice,t_h)$,
    \item $K_{i_{h+1}}\;\!\ba_{h}$ holds at $(\rnice,t_{h+1})$ for every $h<k$. 
\end{itemize}

Clearly, in the nice run $\rnice$ of~$Q$, every process $i_h$ is active at time $t_h$ and $\neg K_{i_h}( \neg\runfact_{nice})$ holds at $(\rnice,t_h)$. Recall that by the choice of $Q$, the protocol $Q$ instructs process $i_h$ to perform $a_h$ at $t_h$ in every run $r$ of~$Q$ such that $\neg K_{i_h}( \neg\nicerun)$ holds at $(r,t_h)$. Consequently, in the nice run of $Q$ the actions are performed and there is no violation of the \OR\ requirements in any run $r$ of~$Q$. It follows that~$Q$ solves the Ordered Response problem, as claimed.
\end{proof}
\end{document}